\documentclass[%
reprint, 
superscriptaddress,
amsmath,amssymb,
aps,
prx,
]{revtex4-2}
\usepackage{graphicx}
\usepackage{epstopdf}
\usepackage{bm}
\usepackage{hyperref}
\usepackage{color}
\usepackage{physics}
\usepackage{amsmath}
\usepackage{tikz}
\usepackage{enumitem}
\usepackage{wrapfig}
\usepackage{bbold}
\usepackage{dsfont}
\usepackage[normalem]{ulem}
\usepackage{amsthm}

\newtheoremstyle{nonumber}
  {}                
  {}                
  {\itshape}        
  {}                
  {\bfseries}       
  {}                
  { }               
  {}                

\theoremstyle{nonumber}
\newtheorem*{theorem}{Theorem}
\newtheorem*{lemma}{Lemma}

\usepackage{changes}

\hypersetup{%
   pdfpagemode=UseNone, 
   pdfstartpage=1,
   pdfmenubar=true,
   pdftoolbar=true,
   colorlinks = true,
   linkcolor=blue,
   citecolor=blue,
   urlcolor=blue,
   bookmarksopen=false}

\definecolor{fuchsia}{rgb}{1.0, 0.0, 1.0}
\definecolor{fuchsia}{rgb}{1.0, 0.0, 1.0}
\definecolor{maroon}{rgb}{0.788, 0.0, 0.086}
\definecolor{ao}{rgb}{0.0, 0.5, 0.0}
\definecolor{green}{rgb}{.2,.6,.2}
\definecolor{brickred}{rgb}{0.8, 0.25, 0.33}
\definecolor{brightcerulean}{rgb}{0.11, 0.67, 0.84}
\definecolor{maroon}{rgb}{0.788, 0.0, 0.086}
\definecolor{ao}{rgb}{0.0, 0.5, 0.0}

\newcommand{\ham}{\hat{H}}

\newcommand{\fig}{Fig.}

\newcommand{\seclab}{Sec.}

\begin{document}

\title{Learning interactions between Rydberg atoms}

\author{Olivier~Simard}
\thanks{These two authors contributed equally.}
\affiliation{Collège de France, PSL University, 11 place Marcelin Berthelot, 75005 Paris, France}
\affiliation{CPHT, CNRS, École Polytechnique, IP Paris, F-91128 Palaiseau, France}
\author{Anna~Dawid}
\thanks{These two authors contributed equally.}
\affiliation{Center for Computational Quantum Physics, Flatiron Institute, 162 Fifth Avenue, New York, NY 10010, USA}
\affiliation{$\langle aQa^L\rangle$ Applied Quantum Algorithms -- Leiden Institute of Advanced Computer Science \\ \& Leiden Institute of Physics, Universiteit Leiden, The Netherlands}
\author{Joseph~Tindall}
\affiliation{Center for Computational Quantum Physics, Flatiron Institute, 162 Fifth Avenue, New York, NY 10010, USA}
\author{Michel~Ferrero}
\affiliation{Collège de France, PSL University, 11 place Marcelin Berthelot, 75005 Paris, France}
\affiliation{CPHT, CNRS, École Polytechnique, IP Paris, F-91128 Palaiseau, France}
\author{Anirvan~M.~Sengupta}
\affiliation{Center for Computational Quantum Physics, Flatiron Institute, 162 Fifth Avenue, New York, NY 10010, USA}
\affiliation{Center for Computational Mathematics, Flatiron Institute, 162 5th Avenue, New York, New York 10010, USA}
\affiliation{Department of Physics and Astronomy, Rutgers University, Piscataway, New Jersey 08854, USA}
\author{Antoine~Georges}
\affiliation{Collège de France, PSL University, 11 place Marcelin Berthelot, 75005 Paris, France}
\affiliation{CPHT, CNRS, École Polytechnique, IP Paris, F-91128 Palaiseau, France}
\affiliation{Center for Computational Quantum Physics, Flatiron Institute, 162 Fifth Avenue, New York, NY 10010, USA}
\affiliation{DQMP, Université de Genève, 24 quai Ernest Ansermet, CH-1211 Genève, Switzerland}

\date{\today}

\begin{abstract}
Quantum simulators have the potential to solve quantum many-body problems that are beyond the reach of classical computers, especially when they feature long-range entanglement. To fulfill their prospects, quantum simulators must be fully controllable, allowing for precise tuning of the microscopic physical parameters that define their implementation. We consider Rydberg-atom arrays, a promising platform for quantum simulations. 
Experimental control of such arrays is limited by the imprecision on the optical tweezers positions when assembling the array, hence introducing uncertainties in the simulated Hamiltonian. In this work, we introduce a scalable approach to Hamiltonian learning using graph neural networks (GNNs). We employ the Density Matrix Renormalization Group (DMRG) to generate ground-state snapshots of the transverse field Ising model realized by the array, for many realizations of the Hamiltonian parameters. Correlation functions reconstructed from these snapshots serve as input data to carry out the training. We demonstrate that our GNN model has a remarkable capacity to extrapolate beyond its training domain, both regarding the size and the shape of the system, yielding an accurate determination of the Hamiltonian parameters with a minimal set of measurements. We prove a theorem establishing a bijective correspondence between the correlation functions and the interaction parameters in the Hamiltonian, which provides a theoretical foundation for our learning algorithm. Our work could open the road to feedback control of the positions of the optical tweezers, hence providing a decisive improvement 
of analog quantum simulators.
\end{abstract}


\maketitle


\section{Introduction}

Analog quantum simulators aim at answering fundamental questions about physics and solve problems that are not accessible to classical computers, such as quantum many-body problems with many degrees of freedom and large-scale entanglement~\cite{bloch2008MBPultracold, Lamata18AdvPhys, Tomza19RevModPhys, Kjaergaard20AnnuRev,Browaeys20NatPhys,Monroe21RevModPhys, baranov2012condensed}. Recently, this promise is becoming a reality thanks to great experimental progress in controlling the individual quantum degrees of freedom of quantum systems with optical tweezers~\cite{EndresScience16, Barredo2016Science} and single atom-resolved detection techniques such as quantum gas microscopy~\cite{Bakr09Nature, Sherson10Nature, Cheuk15PRL, Haller15NatPhys}, among others. Still, experimental quantum platforms are noisy and imperfect, and this impacts quantum computing~\cite{jaksch2000fast, demille2002qcomputation, morgado2021quantum} and sensing~\cite{ludlow2015optical, bongs2019taking}. To recognize and counter these imperfections, the field needs fast and scalable methods to characterize and verify prepared quantum states and Hamiltonians~\cite{Eisert2020qcertification, Carrasco2021qverification}.

One of the frontiers of quantum characterization and verification is Hamiltonian learning, \textit{i.e.} inferring the experimentally realized Hamiltonian through measurements. In realizable cases, the Hamiltonian can be directly reconstructed from knowledge of the steady states of the system \cite{Qi2019singleeigenstate, Bairey2019localmeasurements, Bairey2020Hlearnopensystem, nandy2024reconstructing}, the wavefunction following a quantum quench ~\cite{li2020quantumquench, Olsacher2024HandLlearning} or following the application of single-qubit pulses ~\cite{Wang2015Htomography}. However, the realization of such methods is not always feasible in an arbitrary quantum system due to their reliance on knowledge of the amplitudes of the quantum state.

The Hamiltonian can also be learned without the full knowledge of the prepared quantum state. For instance, in Bayesian techniques, the Hamiltonian parameters are estimated by iteratively updating a probability distribution over possible Hamiltonians based on the system's measurement outcomes~\cite{Bairey2019localmeasurements, evans2019scalablebayesianHL, Landa2022expBayesian}. Meanwhile, variational and gradient-based methods rely on optimizing a cost function based on how well the candidate Hamiltonian predicts the measurement outcomes. These methods can be combined with a trusted quantum simulator if available~\cite{Wiebe2014Hlearning, Wang2017expHlearn}.
The gradient-based approach has been implemented using tensor networks~\cite{wilde2022scalably}, neural differential equations \cite{heightman2024HlearningNDE}, and by estimating the time derivatives of some observables \cite{StilckFranca2024Hestimation, Gu2024practicalHlearning}. 

\begin{figure*}[t]
    \centering
    \includegraphics[width=\textwidth]{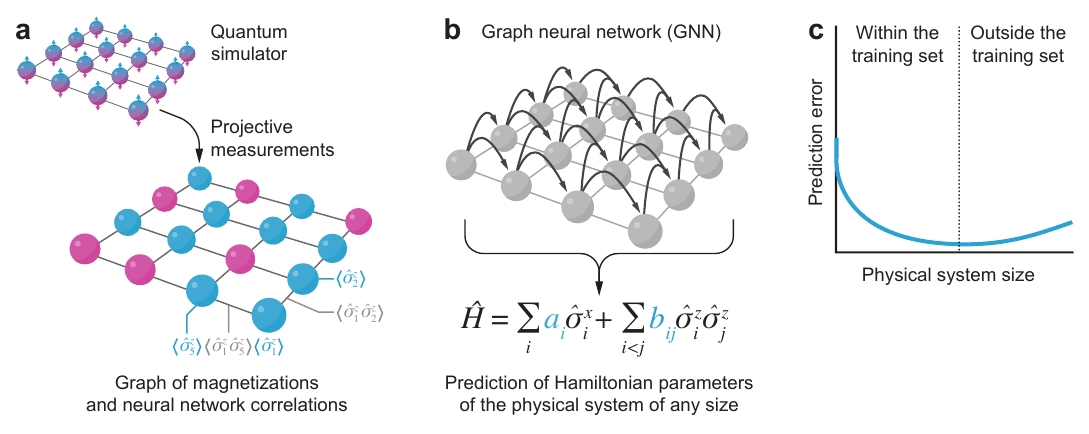}
    \caption{\textbf{Graph neural networks (GNNs) for Hamiltonian learning.} (a) We represent the information about a quantum simulator in the form of a graph, where nodes contain on-site information and edges represent two-body correlators. The sketched graph includes only NN two-body correlators but generally can include correlators of arbitrary range and number of bodies. (b) A GNN takes in a graph of any size and learns to update its values into a useful representation, which then forms a basis for predicting an arbitrarily large number of parameters of pre-defined Hamiltonian terms. (c) We demonstrate that the GNN is able to predict Hamiltonian parameters of quantum systems of large sizes while being trained only on small sizes.}
    \label{fig:intro}
\end{figure*}

Despite these methodological advances, Hamiltonian learning is especially challenging and remains impractical in regimes where \textbf{i}) comparison with classical techniques, such as tensor networks, is no longer possible \cite{nandy2024reconstructing, wilde2022scalably, heightman2024HlearningNDE}, \textbf{ii}) the Hamiltonian has many parameters and terms, some of which may be unknown~\cite{li2020quantumquench, Olsacher2024HandLlearning, wilde2022scalably, Kokail2021EHlearning}, \textbf{iii}) the experimental setup imposes constraints over the accessible measurement set~\cite{li2020quantumquench, Olsacher2024HandLlearning,Wang2015Htomography}, 
and \textbf{iv}) the task needs to be realized fast~\cite{Kokail2021EHlearning} using a viable number of measurements.

Given the ongoing success of machine learning (ML) in quantum physics~\cite{Gebhart2023learningqsystems, dawid2025MLinquantum}, especially in quantum state tomography~\cite{torlai2018QST, carrasquilla2019QST, torlai2019QST, ahmed2021QST, lange2023adaptiveQST, An2024unifiedQSTHL} and phase classification \cite{Nieuwenburg17NatPhys, liu2018discriminativephase, greplova_unsupervised_2020, Kaming2021, patel2022unsupervised, miles2023machine, cybinski2024tetrisCNN}, it is promising to address the outstanding challenges of Hamiltonian learning with neural networks. This direction has already been fruitful~\cite{nandy2024reconstructing, xin2019localmeasurementQST, valenti2019hamlearning, Che2021Hfromsinglequbit, Genois2021, valenti2022hamlearning, koch2023advHlearn, Karjalainen2023Haminference} and outperformed Bayesian methods in a direct comparison~\cite{valenti2022hamlearning}. While these neural-network-based methods are able to quickly predict many Hamiltonian parameters at once with low error based on a user-defined form of measurements, they require accurate numerical simulations of the quantum system under consideration and, therefore, do not scale up to large two-dimensional (2D) systems.

Here, we present a Hamiltonian learning approach that addresses the scalability challenge of neural-network-based approaches and is driven by the needs and limitations of analog quantum simulators based on Rydberg atoms (\fig~\ref{fig:intro}). We demonstrate how a graph neural network (GNN), designed to be invariant with respect to the size of the quantum system, learns Hamiltonian parameters of a quantum transverse-field Ising model (TFIM) defined on small computationally tractable 2D arrays and successfully extrapolates to larger 2D arrays with only small impacts to the GNN's predictive accuracy. This opens the way to tackle experimental challenges such as learning the noise sources or correcting errors in experimentally relevant regimes with the training data obtained via numerical simulations of smaller physical systems. 

In experiments, the sources of imprecision come into two different classes. The first is of a technical nature and originates from the difficulty of making a homogeneously spaced 2D arrangement of Rydberg atoms. The second class encompasses fundamental imprecision sources, which are the quantum zero-point motion of the atoms sitting in the finite potential well and the thermal fluctuations associated with the millikelvin temperatures. Here, we focus on the learning task that addresses the technical uncertainty leading to random optical tweezer displacements in the system. The GNN learns to predict tweezer displacements based on observables such as the local magnetization and nearest-neighbor (NN) and next-nearest-neighbor (NNN) spin correlators, which can be obtained from projective measurements performed on quantum simulators. Thanks to the GNN prediction speed and accuracy, it could be used in future setups to counteract the tweezer displacements, hence providing feedback control on the experiment. 
This would pave the way towards the study of error-free Hamiltonians, such as those featuring quantum spin liquids~\cite{PhysRevB.106.115122}, which would be sensitive to disorder generated by random tweezer displacements due to the narrow parameter ranges where they take place~\cite{doi:10.7566/JPSJ.84.024720,PhysRevB.92.140403,PhysRevB.95.035141}. 

In Sec.~\ref{sec:Methods}, we describe in more detail the ultracold Rydberg platform, which inspired the learning task. There, we also introduce a bijective correspondence between the correlation functions and the interaction parameters of the Hamiltonian, which provides a theoretical foundation for our learning algorithm. This proof is laid out in Appendix~\ref{app:bijection_proof}.  Finally, we describe the learning algorithm itself and introduce the GNN architecture and its training procedure. In \seclab~\ref{sec:results}, we demonstrate the scalability of the GNN and the dependence of its prediction error on data representation. We also show how the GNN prediction error depends on the number of snapshots used to estimate correlators. We furthermore compare the GNN performance to that of a multi-layer perceptron (MLP). In \seclab~\ref{sec:conclusions}, we conclude by discussing open problems and the next directions unlocked by our Hamiltonian learning approach.

\section{Methods}
\label{sec:Methods}

We start this section by describing briefly in Sec.~\ref{sec:physical_Ham} the Rydberg-atom platform for quantum simulation, which inspired our learning task, as well as the Hamiltonian studied in this work. In Sec.~\ref{app:bijection_explained}, we outline a proof, detailed in Appendix~\ref{app:bijection_proof}, that demonstrates a bijective relation between spin exchange couplings and spin-spin correlation functions. This bijective relation ensures, for instance, that NN correlation functions suffice to uniquely determine NN couplings defining the Hamiltonian, and therefore the relative distances of all the NN Rydberg atoms on the lattice. Then, in Sec.~\ref{sec:Machine_learning_technique}, we introduce the key ML techniques used to predict the atomic displacements from the DMRG local magnetization and spin correlators, as well as the metrics used to judge the predictive power of the neural network.

\subsection{Rydberg platform and the studied Hamiltonian}
\label{sec:physical_Ham}

Among the promising quantum simulation platforms stand Rydberg-atom  arrays~\cite{Labuhn2016theoryexp,Scholl2021antiferro2D,Henriet2020quantumcomputing}. By means of optical tweezers, the atoms are trapped and act like pseudo-spin-$\frac12$ qubits, which can be addressed individually. On these platforms, the TFIM can be natively implemented, as well as other spin models, like the quantum XY~\cite{PhysRevLett.114.113002} or quantum XXZ models~\cite{PRXQuantum.3.020303}, by introducing anisotropy in the spin exchange energy. In Rydberg-based platforms, quantum entanglement is generated via the Rydberg blockade mechanism, which arises from the electric dipole-dipole interaction, a second-order process that prevents the simultaneous electronic excitation of nearby atoms into Rydberg states.

In this work, we solve the Hamiltonian learning task for the 2D transverse-field Ising model (TFIM)
\begin{equation}
\label{eq:H}
   \hat{H} = \sum_{i<j} \underbrace{\frac{C_6}{{R_{ij}}^6}}_{\equiv J_{i,j}} \hat{\sigma}^z_i\hat{\sigma}_j^z + \hbar\Omega\sum_i\hat{\sigma}^x_i + \hbar\delta\sum_i\hat{\sigma}^z_i,
\end{equation}
where $i<j$ runs over all atomic positions while avoiding double-counting of pairs, $C_6/\hbar\simeq 5.42\times 10^{6}\,\text{rad}\cdot\mu\text{s}^{-1}\cdot\mu\text{m}^6$ is the dipole-dipole interaction term originating from the van der Waals interaction for Rubidium with principal quantum number $n=70$~\cite{Scholl2021antiferro2D}, with $\hbar$ the reduced Planck constant. The parameter $\Omega$ stands for the transverse Ising field (Rabi frequency), $\delta$ for the detuning field, whilst $\hat{\sigma}^z$ and $\hat{\sigma}^x$ denote the corresponding Pauli matrices. The distance $R_{ij}$ describes the pairwise distance between any two atoms $i$ and $j$ sitting in the array. Hence, the interaction term (first term in Eq.~\eqref{eq:H}) covers all pairs of Rydberg atoms over the whole array, and is in essence a long-range interaction which is responsible for the Rydberg blockade whenever the Rabi frequency $\hbar\Omega\ll C_6/R^6_{ij}\equiv J_{i,j}$. The two-level system is encoded in the Rydberg states, typically $s$ orbitals. In the following, we will set $\hbar = 1$.

To emulate the TFIM on a Rydberg-atom platform, we employ the Density Matrix Renormalization Group (DMRG) algorithm \cite{PhysRevLett.69.2863,PhysRevB.48.10345} to solve the TFIM. We provide more details on the DMRG calculations in Appendix~\ref{app:technical_details_DMRG}.

\subsection{Bijection between interactions and spin-spin correlators}
\label{app:bijection_explained}

This section outlines a general theorem for Hamiltonian reconstruction, established in detail in Appendix~\ref{app:bijection_proof}. This theorem states that any two-body correlation function $c_{i,j}$ that involves an arbitrary pair of lattice (orbital) degrees of freedom is uniquely (one-to-one) mapped to a corresponding Hamiltonian interaction term $J_{i,j}$. It assumes that the correlation matrices $c_{i,j}$ are $J$-representable, meaning they correspond to correlation functions for some choice of $J_{i,j}$, in analogy with the concept of $V$-representability~\cite{PhysRevA.26.1200,ENGLISCH1983253}. This theorem generalizes the Hohenberg-Kohn (HK) theorem from density functional theory~\cite{HohenbergKohn,Garrigue2019HohenbergKohn} and stands as a quantum generalization to the Henderson theorem~\cite{HENDERSON1974197}. In Appendix~\ref{app:bijection_proof}, by means of a lemma, we specialize the proof of the bijection theorem to the transverse field Ising model (TFIM), which holds if and only if $\Omega\neq 0$.

\begin{theorem}[Hohenberg-Kohn-Henderson Theorem for the TFIM]
For the TFIM as defined in Eq.~\eqref{eq:H}, for fixed nonzero $\Omega$ and fixed $\delta$, there exists a bijection between the $J$-representable correlation functions 
$c_{i,j} = \langle\hat{\mathbf{S}}_i^z \hat{\mathbf{S}}_j^z\rangle$ and the set of interactions $J_{i,j}$, 
$J_{i,j} \Leftrightarrow c_{i,j}$.
\end{theorem}

The proof is given in the Appendix, and follows two key steps. First we establish that, for nonzero $\Omega$, two different interaction matrices $J^{(1)}_{i,j}$, $J^{(2)}_{i,j}$ cannot yield identical ground-states $\Psi$, so that $\Psi^{(1)}\neq \Psi^{(2)}$ if $J^{(1)}_{i,j} \neq J^{(2)}_{i,j}$. This allows then to use the variational principle for the energy as a strict inequality to prove through a \textit{reductio ad absurdum} argument that different interactions matrices, $J^{(1)}_{i,j}$, $J^{(2)}_{i,j}$, imply different sets of correlation matrices, $c^{(1)}_{i,j}\neq c^{(2)}_{i,j}$, respectively. This result serves as a theoretical foundation for reconstructing Hamiltonians from experimentally measured correlation functions in quantum systems, like cold atom quantum simulators. 

In the proof, the key constraints come from Eq.~\ref{eq:SamePsiContradiction} which provides a constraint on the interaction differences $\{J^{(1)}_{i,j}-J^{(2)}_{i,j}\}$ for each conventional $\hat{\sigma}^z$-basis spin configuration $\ket{\sigma}$, for which the amplitude is non-zero. If we randomly picked such a $\ket{\sigma}$ configurations, we would need more than $N(N-1)/2+1$ such configurations to show that these differences are zero. In the context of TFIM with $\Omega\neq 0$, we have an overabundance of such constraints, namely one for each $2^N$ basis vectors, thanks to the lemma detailed in the Appendix~\ref{app:bijection_proof}. We therefore think that, beyond the TFIM, a similar theorem could be proven for systems without strong conditions like stoquasticity \cite{bravyi2008complexity} or irreducibility \cite{crosson2017quantum} for the density matrix (properties enjoyed by the $\Omega\neq0$ TFIM), provided two distinct sets of interactions cannot yield the same ground-state.

\begin{figure*}[t]
    \centering
    \includegraphics[width=\textwidth]{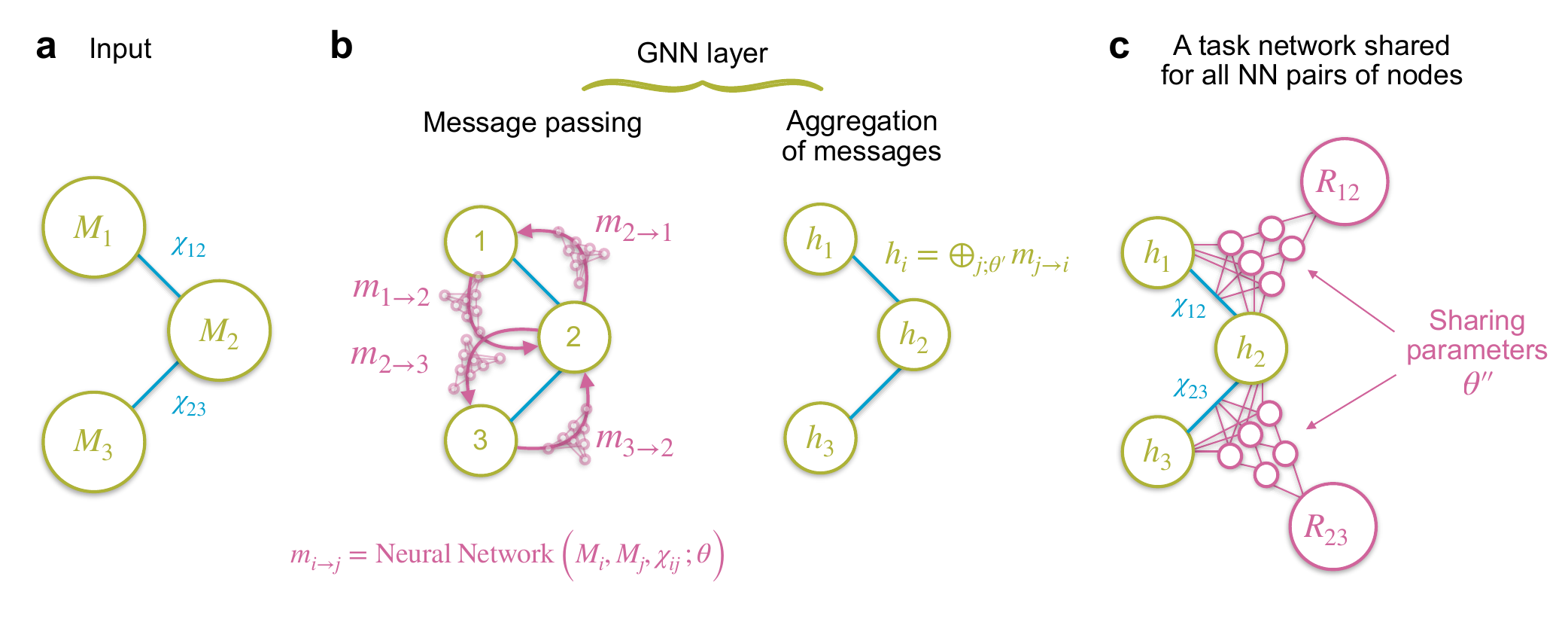}
    \caption{\textbf{Illustration of the elements of the GNN architecture} on an example of a three-atom one-dimensional chain. (a) Input is a graph containing spin correlators measured in a quantum system. Here, for example, nodes contain on-site magnetization, $M_i$, where $i$ indicates a site and nearest-neighbor edges connecting nodes $i$ and $j$ contain NN spin-spin correlations, $\chi_{ij}$. Graphs can contain higher-order many-body and longer-range correlators. (b) A GNN layer updates each node value with a new hidden value $h_i$, preserving the graph structure. The update is based on messages sent between the neighboring nodes, $m_{i \rightarrow j}$, which is predicted by the same message-passing network that takes in the values of neighboring nodes and edges. Then the messages are aggregated via a learned combination of aggregators such as min, max, and mean, and form a new hidden value of each node. (c) The same task network takes in every pair of neighboring nodes along with their spin-spin correlation and predicts the distance between them. With such an architecture, the GNN is invariant to the physical size input and bases its predictions on the atom neighborhood. When comparing to the MLP model, we solely use the very same task networks to read off the NN correlators and infer the NN relative distances.}
    \label{fig:GNN}
\end{figure*}

\subsection{Machine learning techniques}
\label{sec:Machine_learning_technique}

As previously mentioned, the physical setup of interest consists of Rydberg atoms arranged on a 2D surface and maintained in place with optical tweezers. The measurement outcomes of this system, in terms of magnetization and spin-spin correlators, can be natively represented as a graph, see \fig~\ref{fig:GNN}(a). For such a data representation, a well-suited ML architecture comes from the wide class of graph neural networks (GNN). The advantage that is especially important in the context of this work is that they preserve the graph structure of the processed data and can be made size invariant. 
In our case, we achieve size-invariance by combining the GNN with properly designed \textit{task} networks. Specifically, we employ the Principal Neighbor Aggregation (PNA) variant with multiple aggregators~\cite{corso2020principalneighbourhoodaggregationgraph} as a GNN layer, and multi-layer perceptrons (MLPs) as task networks. We present some elements of the GNN in \fig~\ref{fig:GNN}, and explain the key elements defining the GNN and how it processes the physical data in \seclab~\ref{sec:GNN}. To train the GNN, the $L^2$ norm between targets $\mathbf{y}$ and predictions $\hat{\mathbf{y}}$ is evaluated, where the targets and predictions correspond to relative nearest-neighbor (NN) distances ($R_{ij} \to R_{\langle ij\rangle}$) and relative next-nearest-neighbor (NNN) distances ($R_{ij} \to R_{\langle\langle ij\rangle\rangle}$) between atoms
\begin{align}
\label{eq:L2_minimization}
\text{MSE}(\mathbf{y},\hat{\mathbf{y}}) = \|\mathbf{y}-\hat{\mathbf{y}}\|^2,
\end{align}
with $\|\cdot\|$ denoting the squared $L^2$ modulus of the vector. Thus, Eq.~\eqref{eq:L2_minimization} is the cost function to be minimized via gradient descent techniques. The optimizer used during training is AdamW~\cite{loshchilov2019decoupledweightdecayregularization}.

\subsubsection{Graph neural network}
\label{sec:GNN}

A GNN is a neural network architecture that processes data represented as a graph \cite{battaglia2018GNNs}. This architecture is especially attractive for tackling physical problems as the GNN processes the data in a physics-inspired way, emulating interactions between nodes in a graph. Moreover, the GNN layer is invariant to the input size and has more inherent interpretability than other types of neural networks \cite{Cranmer2020}. It has been widely used in classical physics \cite{Shlomi2021GNNparticlephysics, cranmer2021disentangled, Lemos2023rediscovering} and is slowly entering quantum physics as well \cite{tran2022using}.

The GNN layer's input is a graph, and its output is a processed graph of the same structure but modified node (and possibly edge) values. In \fig~\ref{fig:GNN}(a), we show an exemplary input graph that is relevant to this work and that describes the outcomes of measurements on a three-spin chain, as discussed in Secs.~\ref{sec:physical_inputs} and \ref{sec:shemes_correlation_function_estimation}. The GNN layer processes the input graph by sending learnable messages $m_{i \rightarrow j}$ between all pairs of directly connected nodes $i$ and $j$, as shown in \fig~\ref{fig:GNN}(b). All messages within one GNN layer are outputs of the same neural network whose parameters $\mathbf{\theta}$ are optimized to send messages that are relevant to the task at hand. The input of this message-passing neural network is the values of the node pair between which the message is sent (here, $M_i$ and $M_j$) and the value of the edge connecting the node pair (here, $\chi_{ij}$). The same neural network with the same $\mathbf{\theta}$ can send different messages between various pairs of nodes because its input ($M_i$, $M_j$, and $\chi_{ij}$) varies between them. In our case, the message-passing network is a fully-connected neural network. The architectural details are in Appendix~\ref{app:technical_details_training_testing}.

After the messages are sent, they are aggregated in each node to form a new updated hidden value of a node, $h_i = \oplus_{j; \mathbf{\theta'}} \, m_{j \rightarrow i}$. In this work, we employ PNA~\cite{corso2020principalneighbourhoodaggregationgraph} as the GNN layer architecture. Instead of using one aggregation operation, PNA uses a learnable linear combination $\oplus_{\mathbf{\theta'}}$ of multiple aggregators such as summation, multiplication, max, min, and average. Note that it is at the aggregation stage that edge effects in the system can be taken into account. The described GNN layer, composed of message passing and aggregation, can be repeated several times in a network architecture, effectively allowing for `communication' between nodes that are further apart from each other. In this work, we use four consecutive GNN layers, each composed of message passing and aggregation of messages.

We design the full neural network architecture to be independent of the size of the physical system. While the GNN layer natively has this property, ultimately, we also need a \textit{task} network that makes predictions based on the hidden representation of the graph learned by the GNN layer(s). Therefore, we combine the GNN layers with a neural network parametrized by $\mathbf{\theta''}$ that takes as an input hidden values of only two NN nodes $h_i$ and $h_j$ along with their spin-spin correlation $\chi_{ij}$ and predicts the distance between NN spins $i$ and $j$, $R_{\langle ij\rangle}$ (or $R_{\langle\langle ij\rangle\rangle}$ in case of NNN neighbor spins), as presented in \fig~\ref{fig:GNN}(c). The input size independence is achieved here in a similar way as in the GNN layer. One task network learns a mapping between an NN pair and a distance between them and then makes predictions for all possible NN pairs. It can give different predictions for various NN pairs because the hidden values of participating nodes serving as the neural network input are different. In our case, the task network is a fully-connected neural network with two hidden layers. More generally, task networks can be adjusted to the form of Hamiltonian terms whose parameters they need to learn. For one-body Hamiltonian terms, the task network can take one $h_i$ as an input. For $n$-body terms, the input can include hidden values of $n$ nodes.

Therefore, the training of the full neural network aims at finding such parameters $\mathbf{\Theta} \in \{ \mathbf{\theta}, \mathbf{\theta'}, \mathbf{\theta''} \}$ that allow sending useful messages between the graph nodes, an optimal message aggregation that takes into account edge effects in the system, and finally a correct prediction of distances between NN spins, based on the developed hidden representation of the graph.

In \seclab~\ref{sec:comparison_GNN_MLP}, we compare the GNN performance to that of an MLP. When doing so, we employ the same task networks' architecture but without the GNN processing of the graph -- each task network directly takes the NN spin-spin correlation function of an atom pair as input and predicts the corresponding NN distance as output.

\subsubsection{Physical inputs}
\label{sec:physical_inputs}

The input layer to the GNN is fed with samples built in the following way. On the one hand, the node features consist of local magnetization expectations $M_i=\langle\hat{\sigma}^z_i\rangle_{0;\vec{\Omega}}$, with $\hat{\sigma}^z_i$ the Pauli diagonal operator at site $i$ (in the Z-basis), for a sequence of transverse field $\vec{\Omega}$~\cite{dwivedi2022benchmarkinggraphneuralnetworks}. On the other hand, the edge features consist of NN spin-spin correlation functions ($\chi^{\text{NN}}$) and NNN spin-spin correlation functions ($\chi^{\text{NNN}}$), defined as
\begin{align}
\label{eq:NN_spin_corr}
\chi^{z(x)}_{i,j; \vec{\Omega}} = \langle \hat{\sigma}^{z(x)}_i\hat{\sigma}^{z(x)}_j\rangle_{0;\vec{\Omega}} \propto c_{i,j;\vec{\Omega}},
\end{align}
where $i,j$ are NN and NNN connections between Rydberg atoms, respectively, and $\langle\cdot\rangle_0$ means that the expectation value is carried out in the ground state of the system. The vector notation $\vec{\Omega}$ used as subscript means that samples are arrays of such correlation functions across different equally-spaced values of the transverse field (we assume the field $\Omega$ is spatially homogeneous, \textit{i.e.} the same for all atomic sites).

\subsubsection{Schemes to estimate correlation functions}
\label{sec:shemes_correlation_function_estimation}

All the physical observables considered, namely the magnetization and spin-spin correlation functions, are computed using DMRG~\cite{PhysRevLett.69.2863,PhysRevB.48.10345}. There are two distinct ways to evaluate observables from a converged Matrix Product State (MPS) wavefunction obtained via DMRG: exactly via direct contraction of the tensor network $\langle \psi \vert O \vert \psi \rangle$ or with sampling by taking $N_{\rm sample}$ bitstring samples from the MPS distribution $p(x) \sim \vert \psi(x) \vert^{2}$, which can be done efficiently and perfectly ~\cite{stoudenmire2010mpssampling, ferris2012mpssampling}. We refer to the datasets obtained with these two methods as \textit{exact} or \textit{snapshot} datasets, respectively. The former exact datasets are used unless specified otherwise.

When computing the snapshot observables, one can sample the MPS expressed in the Z- or X-basis and then measure in the Z-basis to compute $\chi^{z}$ or $\chi^{x}$, respectively. Indeed, from Eq.~\eqref{eq:NN_spin_corr}, 

\begin{align}
\label{eq:NN_spin_corr_X}
\chi^{x}_{i,j; \vec{\Omega}} = \sum_{\sigma} \left|\bra{\sigma}\hat{H}_{\text{Hd}}\ket{\Psi}_{\vec{\Omega}}\right|^2 \bra{\sigma}\hat{\sigma}^z_i\hat{\sigma}^z_j\ket{\sigma},
\end{align}
where $\hat{H}_{\text{Hd}}$ is the involutive Hadamard operator defined in the Z-basis 
\begin{align*}
\hat{H}_{\text{Hd}} \equiv \frac{1}{\sqrt{2}}\begin{pmatrix}
1 & 1\\
1 & -1
\end{pmatrix}.
\end{align*}
Again, in Eq.~\eqref{eq:NN_spin_corr_X}, the set $\{\ket{\sigma}\}$ denotes the eigenstates of the $\hat{\sigma}_k^{z}$'s.

In the snapshot scenario where we sample the MPS to compute observables, we resort to Eq.~\eqref{eq:NN_spin_corr_X} to compute $\chi^{x}$. We combine the Z-observable ($\chi^{z}$) and X-observable ($\chi^{x}$) in the GNN training at some explicit occasions (the specific training scenarios are covered in Sec.~\ref{sec:training_scenarios}). Moreover, by default, if not explicitly stated otherwise, the inputted observables are Z-observables. The technical details related to the training and testing of the GNN are discussed further in Appendix~\ref{app:technical_details_training_testing}.

\subsubsection{Training scenarios}
\label{sec:training_scenarios}

We have devised several training scenarios in this work to compare different types of input data and architectures and select the best learning procedure. These training scenarios are listed in Table~\ref{table:framed} and described in the following, in order. To start with, \textbf{i}) the local magnetization $\langle\hat{\sigma}_i^z\rangle_{0;\vec{\Omega}}$, piled up across $\vec{\Omega}$, is inputted to the GNN without providing any physical information to the edges of the graph that stores the information about the cold atom array. Only the disorderless edge distances are provided as edge features. Then, \textbf{ii}) in addition to the local magnetization, the NN spin correlation functions (Eq.~\eqref{eq:NN_spin_corr}) stacked up along $\vec{\Omega}$ are inputted to the GNN input layer to dress up the NN graph edges. This scenario should substantially improve the performance of the neural network according to the bijective relation between connected spin correlators $c_{i,j}$ and spin couplings $J_{i,j}$ (see Sec.~\ref{app:bijection_proof} for the proof). Next, \textbf{iii}) we add the NNN spin correlation functions $\chi^{\text{NNN}}$ to the input data, thereby supplementing the NNN graph edges with physical information. In order to evaluate and discriminate the physical benefits from the architectural ones, we also consider scenarios where the features of NNN graph edges are set to unity to help the flow of information during optimization. This way, the NNN edges do not carry any information about $\chi^{\text{NNN}}$. Hence, \textbf{iv}) the features of the NNN graph edges are set to unity across $\vec{\Omega}$ to check out the effect of including NNN edges in the GNN architecture without assigning them any physical information (no values of $\chi^{\text{NNN}}$). The NNN edges \textit{per se} could act as `skip connections' in the back-propagation of the gradients. To continue the comparisons, in light of the bijection relation covered in Sec.~\ref{app:bijection_proof}, \textbf{v}) we remove the physical content from the NNN graph edges (no values of $\chi^{\text{NNN}}$) and graph nodes (no values of local magnetization $M$), and set the features to unity across $\vec{\Omega}$. To finish the comparisons, \textbf{vi}) we combine the Z- and X-observable measurements of the NN and NNN correlators (Eq.~\eqref{eq:NN_spin_corr}) in the training dataset. When computing X-observables by sampling the MPS, we resort to Eq.~\eqref{eq:NN_spin_corr_X}. The magnetization along $z$ is also considered in the training, and a $\Omega$-history is used. All the training scenarios hitherto mentioned are summarized in order in Table~\ref{table:framed}, and they all use the same $\Omega$-history (see further down below).

\begin{table}[h!]
\centering
\begin{tabular}{||l|l||}  
\hline  
Case $\#1$ & $M$ and $\Omega$-history \\ \hline
Case $\#2$ & $M$, $\chi^{\text{NN}}$ and $\Omega$-history \\ \hline
Case $\#3$ & $M$, $\chi^{\text{NN}}$, $\chi^{\text{NNN}}$ and $\Omega$-history \\ \hline
Case $\#4$ & $M$, $\chi^{\text{NN}}$, NNN edges and $\Omega$-history \\ \hline
Case $\#5$ & $\chi^{\text{NN}}$, NNN edges and $\Omega$-history \\ \hline
Case $\#6$ & $M$, $\chi^{\text{NN}}$, $\chi^{\text{NNN}}$ and $\Omega$-history (Z+X) \\ \hline
\end{tabular}
\caption{\textbf{Table of training scenarios.} Summary of the training cases used in this work. If not explicitly mentioned in a particular scenario, all correlation functions are measured in the Z-basis. As a reminder, $M$ represents the local magnetization.}
\label{table:framed}
\end{table}

To generate the graph dataset fed into the GNN, we choose a nominal NN distance between atomic sites $\bar{R}_{\langle i,j\rangle}$ of 10 $\mu$m, corresponding to the array spacing $a$~\cite{Browaeys20NatPhys} -- the notations $a$ and $\bar{R}_{\langle i,j\rangle}$ will be used interchangeably to denote the nominal NN distance. The positions of the Rydberg atoms are slightly shifted about their nominal values, with perturbations comprised within the bracket $\Delta S \in \left[-0.1,0.1\right] \ \mu$m, corresponding to the width of the white box shown in the phase diagram in Fig.~\ref{fig:phase_diag_resolved_5x5}. $\Delta S$ is drawn from a uniform distribution. Hence, the NN relative atomic distance $R_{\langle i,j\rangle} = \bar{R}_{\langle i,j\rangle} + \Delta R_{\langle i,j\rangle}$, where the perturbation in the NN relative distance between atoms lies within $\Delta R_{\langle i,j\rangle} \in \left[-0.2,0.2\right] \ \mu$m. The white box in Fig.~\ref{fig:phase_diag_resolved_5x5} spans the allowed values of the transverse field $\Omega$ selected during the training phase. 

Component-wise, the GNN edge-feature space is made of copies of the exchange interactions $J_{i,j}$ at different equally-spaced values of $\Omega$. The explicit array of ten $\Omega$'s, used by all the scenarios laid out in Table~\ref{table:framed}, expands as such

\begin{align*}
\vec{\Omega} = \left[-\frac{900}{9},-\frac{700}{9},\cdots,\frac{700}{9},\frac{900}{9}\right] \ \text{rad}\cdot\mu s^{-1},
\end{align*}
and is enclosed in the white box of Fig.~\ref{fig:phase_diag_resolved_5x5}. In the definition of $\vec{\Omega}$, an arbitrary increment of $200/9 \ \text{rad}\cdot\mu s^{-1}$ separates each consecutive elements, so as to span both phases of the phase diagram displayed in \fig~\ref{fig:phase_diag_resolved_5x5}.

\subsubsection{Metrics}
\label{subsubsec:metrics}

While investigating the scenarios listed above in Sec.~\ref{sec:training_scenarios}, Table~\ref{table:framed}, we use three main metrics to evaluate performance. Throughout this work, the figure layout orders the metrics from top to bottom. Hence, we first plot in figures the coefficient of determination $R^2$, defined as 
\begin{align}
\label{eq:equation_coefficient_determination}
R^2 = 1 - \frac{\sum_i(y_i-\hat{y}_i)^2}{\sum_i(y_i-\bar{y}_i)^2},
\end{align}
which qualifies the predictive power of the statistical ML process. In Eq.~\eqref{eq:equation_coefficient_determination}, $\bar{y}_i$ stands for the mean of the images $y_i$ of a function to be modeled, and $\hat{y}_i$ represents the prediction of the model, associated to $y_i$. When $R^2$ approaches $1$ from below, the model is deemed to effectively capture the key characteristics of an unknown underlying function. Next, we consider the mean absolute error (MAE), defined as
\begin{align}
\label{eq:mean_absolute_error}
\text{MAE} = \sum_{i=1}^N \frac{\left|\hat{y}_i-y_i\right|}{N},
\end{align}
whose dimension is in nanometers. In Eq.~\eqref{eq:mean_absolute_error}, the sample size is denoted by $N$, and it grows larger as the cluster size increases since the number of graph edges and nodes increases as well. Recall that the extrema of the possible variations (perturbations) of the NN relative distance is $\Delta R = \pm 0.2 \ \mu$m. Then, at last, the median of the absolute error, denoted MEDAE for short, is plotted. The combined estimation of MEDAE and MAE allows us to evaluate the symmetry of the learned probability distribution: if it is skewed, the MEDAE and MAE differ significantly, whereas if it is symmetric, both the MAE and MEDAE look very similar.

\section{Results}
\label{sec:results}

As pointed out in Sec.~\ref{sec:shemes_correlation_function_estimation}, there are two distinct ways to estimate the physical observables from the DMRG ground state. One way is to work out the expectation values having access to the full wave function content, thereby generating exact training datasets, while the other is to `sample' the wave function thereby producing measurement snapshots constituting snapshot datasets. The former is useful to test thoroughly and faithfully the limitations and strengths of the GNN architecture employed, leaving out extra sources of statistical noises that would come from the projective measurements -- only the numerical round-off error would be the source of `noise'. Therefore, the extent to which the bijection relation proven in Appendix~\ref{app:bijection_proof} is respected by the PNA GNN can be assessed. The performance of the GNN training using this exact scheme is presented in Sec.~\ref{subsec:predictions_from_exact_wave_function}. The GNN training scheme utilizing snapshot observables, in closer relation to experiments, is presented in Sec.~\ref{sec:predictions_snapshots}. These two ways of measuring observables should coincide perfectly with each other in the limit where the number of snapshots tends to infinity. Before delving into GNN Hamiltonian learning results of Secs.~\ref{subsec:predictions_from_exact_wave_function} and \ref{sec:predictions_snapshots}, we ought to study more closely the phase diagram of the TFIM at the constant zero detuning field ($\delta = 0$ in Eq.~\eqref{eq:H}). This is done in Sec.~\ref{subsec:phase_diagram}.

\subsection{Phase diagram}
\label{subsec:phase_diagram}

In Fig.~\ref{fig:phase_diag_resolved_5x5}, we illustrate the phase diagram of the TFIM for a window of nominal NN distances $a=\left[7,12\right] \ \mu m$. Two distinct magnetic orders feature on the phase diagram, namely the antiferromagnetic phase (AFM) and the ferromagnetic phase (FM). The color plot displays the degree to which spins are ordered successively in an anti-parallel fashion, namely the magnetization $M$:

\begin{align}
\label{eq:magnetization_square_lattice}
M = \frac{1}{N}\sum_{i} (-1)^{x_i+y_i}\langle \hat{\sigma}^z_i\rangle_0,
\end{align}
where $N$ represents the total number of array sites, $x_i$ and $y_i$ are integers accounting for the Cartesian components of the orthonormal Bravais lattice vectors of a square lattice. The white box appearing on Fig.~\ref{fig:phase_diag_resolved_5x5} encloses the parameter region of the Hamiltonian \eqref{eq:H}, from which are drawn the physical parameters, mapped to graphs, that will constitute the training dataset of the GNN. The region was selected to span over domains of distances and transverse fields representative of the experimental setups. To give ideas of orders of magnitude, considering that $\bar{R}_{\langle i,j\rangle}\simeq 10\mu\text{m}$, the interaction coefficient showing up in Eq.~\eqref{eq:H} is $C_6/\bar{R}^6_{\langle i,j\rangle}\sim 10 \ \text{rad}\cdot \mu s^{-1}$, hinting on the whereabouts of the phase transition. By defining $\Omega_c$ the critical Rabi frequency above which the magnetization \eqref{eq:magnetization_square_lattice} falls below $10^{-1}$, for $\bar{R}_{\langle i,j\rangle}\simeq 10\mu\text{m}$, we obtain $\Omega_c\simeq 17 \ \text{rad}\cdot \mu s^{-1}$ (the phase diagram is symmetric about $\Omega=0$).

\begin{figure}[h!]
  \centering
    \includegraphics[width=0.95 \columnwidth]{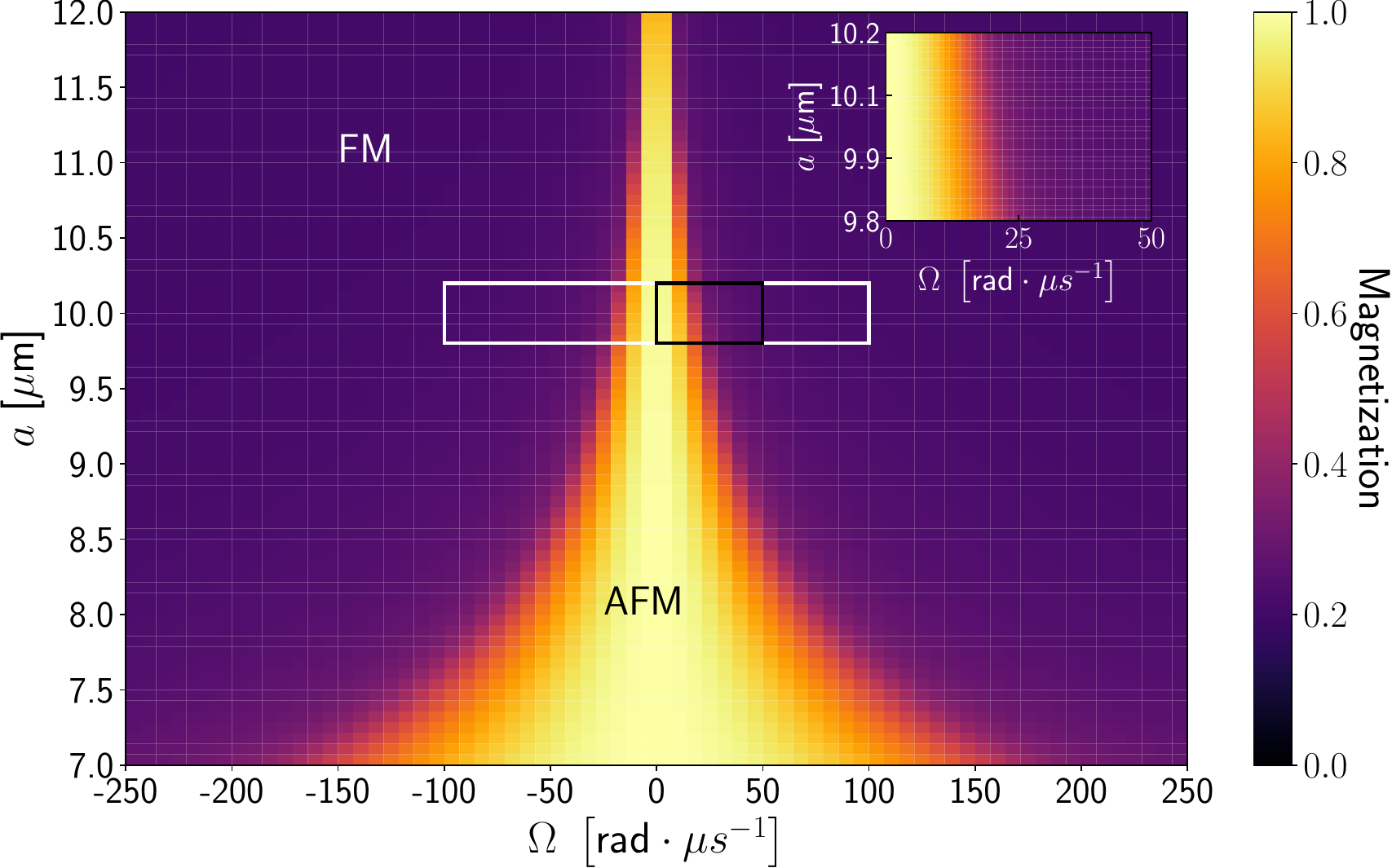}
      \caption{\textbf{Phase diagram of the Rydberg TFIM} for a $5\times 5$ square lattice. The white box encloses the parameter sets that can possibly be drawn during the training of the GNN. Along the $y$-axis, $a$ represents the NN distance between atoms. The inset plot zooms into the region of the phase diagram enclosed within the black box. Note that the phase diagram was calculated for an equally-spaced ordered lattice.}
  \label{fig:phase_diag_resolved_5x5}
\end{figure}

We then zoom in on the vicinity of the phase transition in Fig.~\ref{fig:phase_diag_resolved_5x5} (inset plot) to see better the variation of the magnetization. The range of nominal NN distances $a$ corresponds to the height of the white box appearing in Fig.~\ref{fig:phase_diag_resolved_5x5}. Note that its height encompasses all the possible $a$'s given that $\Delta R_{\langle i,j\rangle} \in \left[-0.2,0.2\right]\mu$m. In the region with a significant color gradient, the variance of the NN spin correlation functions is highest, as will be seen later in Fig.~\ref{fig:NN_corr_variance_delta}. $\Omega \gtrsim 10\,\text{rad}\cdot\mu s^{-1}$ corresponds to the point where the Rabi frequency and interaction terms are of the same order of magnitude, and thus, there is more variability in the ground state.

\subsection{Predictions from the exact wave function}
\label{subsec:predictions_from_exact_wave_function}

In Sec.~\ref{sec:discriminating_between_training_scenarios}, we show the scalability power of the GNN in the Hamiltonian learning task. We also resort to the metrics introduced in Sec.~\ref{subsubsec:metrics} to compare against each other the various GNN training scenarios laid out in Table.~\ref{table:framed}. In particular, we dress up the graph inputs to the GNN with exact physical correlation functions estimated from the MPS, encoding the full ground-state distribution function via its amplitudes; thus throughout this section, the GNN is trained over exact training datasets. Furthermore, in Sec.~\ref{sec:effect_Omega_histories}, we study how the choice of the $\Omega$ impacts the predictive and extrapolation power of the GNN. Finally, in Sec.~\ref{sec:additional_studies}, we summarize the results of additional studies from Appendices~\ref{app:amount_training_datasets_training_datasets}-\ref{app:target_optimization}.

\subsubsection{GNN scalability and comparison between training scenarios}
\label{sec:discriminating_between_training_scenarios}

We illustrate in Fig.~\ref{fig:metrics_training_GNN} the indicative metrics for all training scenarios covered in Sec.~\ref{sec:Machine_learning_technique} (see Table.~\ref{table:framed}), following through the figure layout described in Sec.~\ref{subsubsec:metrics}. The error bars show the standard error on the GNN metric results based on a set of five independently trained GNNs within each case scenario. The main takeaway is that in every scenario the GNN exhibits an impressive extrapolation ability to larger physical system sizes than the ones included in the training dataset. In particular, we trained here the GNN on data coming from physical system sizes of $\{4\times 4,5\times 5,6\times 6\}$ and the GNN successfully extrapolated to sizes up to $\{9\times 9\}$ (which constituted a numerical limit for our DMRG calculations) with only a small impact on its predictive error. Its successful extrapolation to rectangular arrays is also notable.

From looking at Fig.~\ref{fig:metrics_training_GNN}, it seems like the GNN does not benefit from physical information provided by the various correlations functions, namely the NNN spin correlation function and local magnetization, as much as from the architectural add-ons like the NNN edges; it is even more obvious when the local magnetization is set to $1$ (case $\#5$ in Table~\ref{table:framed}). Nonetheless, including the physical information coming from $\chi^{\text{NNN}}$ is not without benefits, as can be seen from comparing the case $\#3$ (extra $\chi^{\text{NNN}}$) with $\#2$ since the former case outperforms the latter as far as MAE~\eqref{eq:mean_absolute_error} and MEDAE are concerned. By very far, the worst GNN training implementation is case $\#1$, where absolutely no physical edge information -- no NN nor NNN correlation function -- is included in the graphs of the dataset. When adding the X-correlators (Eq.~\eqref{eq:NN_spin_corr_X}) like done in case $\#6$, the prediction power of the model can be increased further, amassing more physical input through the training, at the cost of a higher standard error on the extrapolation of large cluster sizes.

\begin{figure}
  \centering
    \includegraphics[width=0.95\columnwidth]{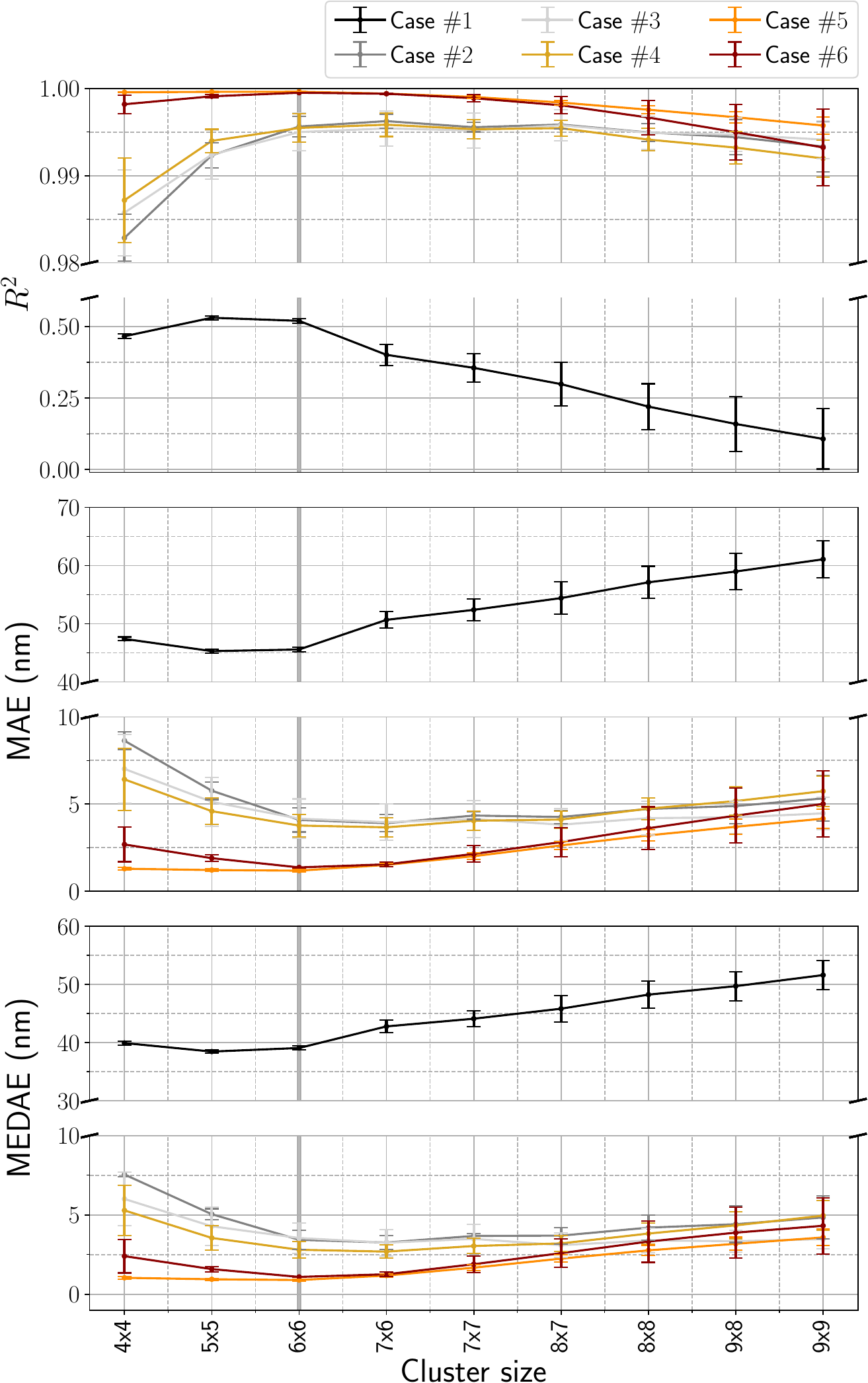}
      \caption{\textbf{Performance of the GNN training scenarios.} Top panel: $R^2$ for various cluster sizes in the test dataset. Middle panel: Mean absolute error of predictions. Bottom panel: Median of the absolute error of predictions vs targets. The vertical grey line indicates the cluster size beyond which the GNN extrapolates from its training dataset. Note that all cases use the same $\Omega$-histories $\vec{\Omega}$. The cases are detailed in Table~\ref{table:framed}. The error bars show the standard error on five distinct and independent trainings.}
  \label{fig:metrics_training_GNN}
\end{figure}

Moreover, it looks like the additional physical information to $\chi^{\text{NN}}$ is detrimental to the predictions of the GNN for the cluster sizes part of the training dataset (before the grey vertical line); this can be seen in all three measures shown in Fig.~\ref{fig:metrics_training_GNN}. This extra information also seems to negatively impact the error for large cluster sizes: the smallest error is obtained using the case $\#5$ while the largest is obtained using the case $\#6$. The GNN is also more faithful to the bijection relation described in Sec.~\ref{app:bijection_proof} when no physical observables are supplemented to NNN graph edges or nodes. Indeed, compared to case $\#5$, both case $\#4$, corresponding to the addition of local magnetizations to nodes, and case $\#3$, corresponding to the additions of local magnetizations to nodes and $\chi^{\text{NNN}}$ to NNN graph edges, diminish the performance.

\begin{figure}[t]
  \centering
    \includegraphics[width=0.95\columnwidth]{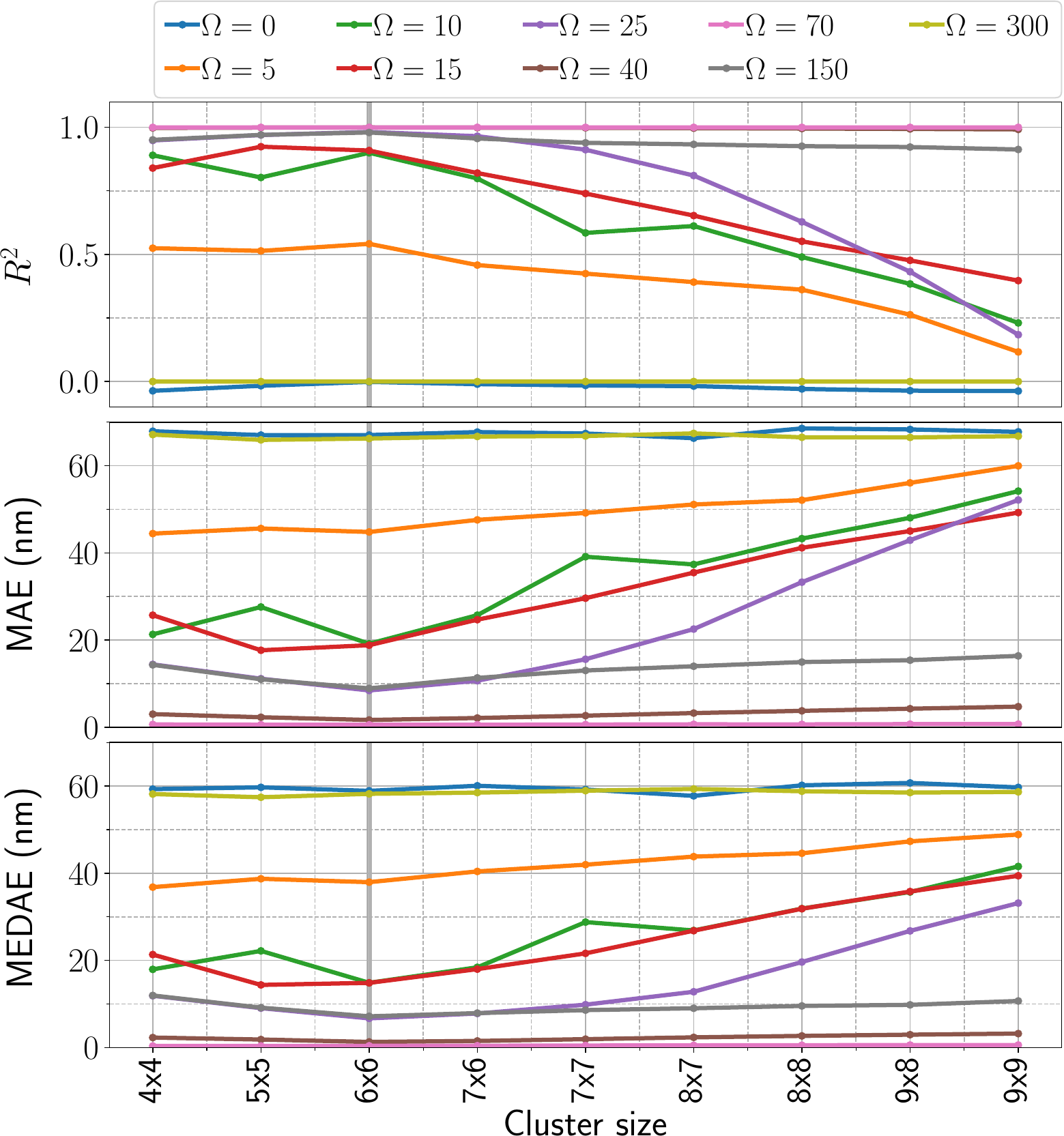}
      \caption{\textbf{$\Omega$-resolved performance of the training case $\#3$.} Top panel: $R^2$ as a function of Rydberg cluster sizes at various values of $\Omega$ (see legend). Middle panel: Mean absolute error of the predicted nearest-neighbor (NN) relative displacements. Bottom panel: Median of the mean absolute error of the relative NN displacements. $\Omega$ is shown in units of $\text{rad}\cdot \mu s^{-1}$.}
  \label{fig:metrics_deltas}
\end{figure}

\subsubsection{Effect of $\vec{\Omega}$-histories}
\label{sec:effect_Omega_histories}

In the training scenarios covered in Fig.~\ref{fig:metrics_training_GNN}, we only considered cases where histories of $\Omega$ are used to enrich both the edge and node features in the graphs. Of course, one can wonder what happens if one selects solely one value of $\Omega$ and if some $\Omega$'s have higher information content for the Hamiltonian learning task than the others. For that matter, we show in Fig.~\ref{fig:metrics_deltas} a collection of single $\Omega$ used to train the GNN following case $\#3$, in that the physical observables like $M_i$, $\chi^{\text{NN}}$ and $\chi^{\text{NNN}}$ are used as node and edge attributes, but the dataset is limited to only one $\Omega$ instead of a $\Omega$-history. From looking at Fig.~\ref{fig:metrics_deltas}, it is striking to see how the training performance hinges on the selected value of $\Omega$. We nevertheless can confidently state that the values of $\Omega$ chosen above $\Omega\simeq C_6/\bar{R}^6_{\langle i,j\rangle}\sim 10 \ \text{rad}\cdot \mu s^{-1}$, although not too far away, give significantly better results. Furthermore, above that threshold, in the inset plot of Fig.~\ref{fig:phase_diag_resolved_5x5}, for a fixed value of $\Omega$, the magnetization $M$~\eqref{eq:magnetization_square_lattice} varies more within the same set of $a$'s.

\begin{figure}[t]
  \centering
    \includegraphics[width=0.95\columnwidth]{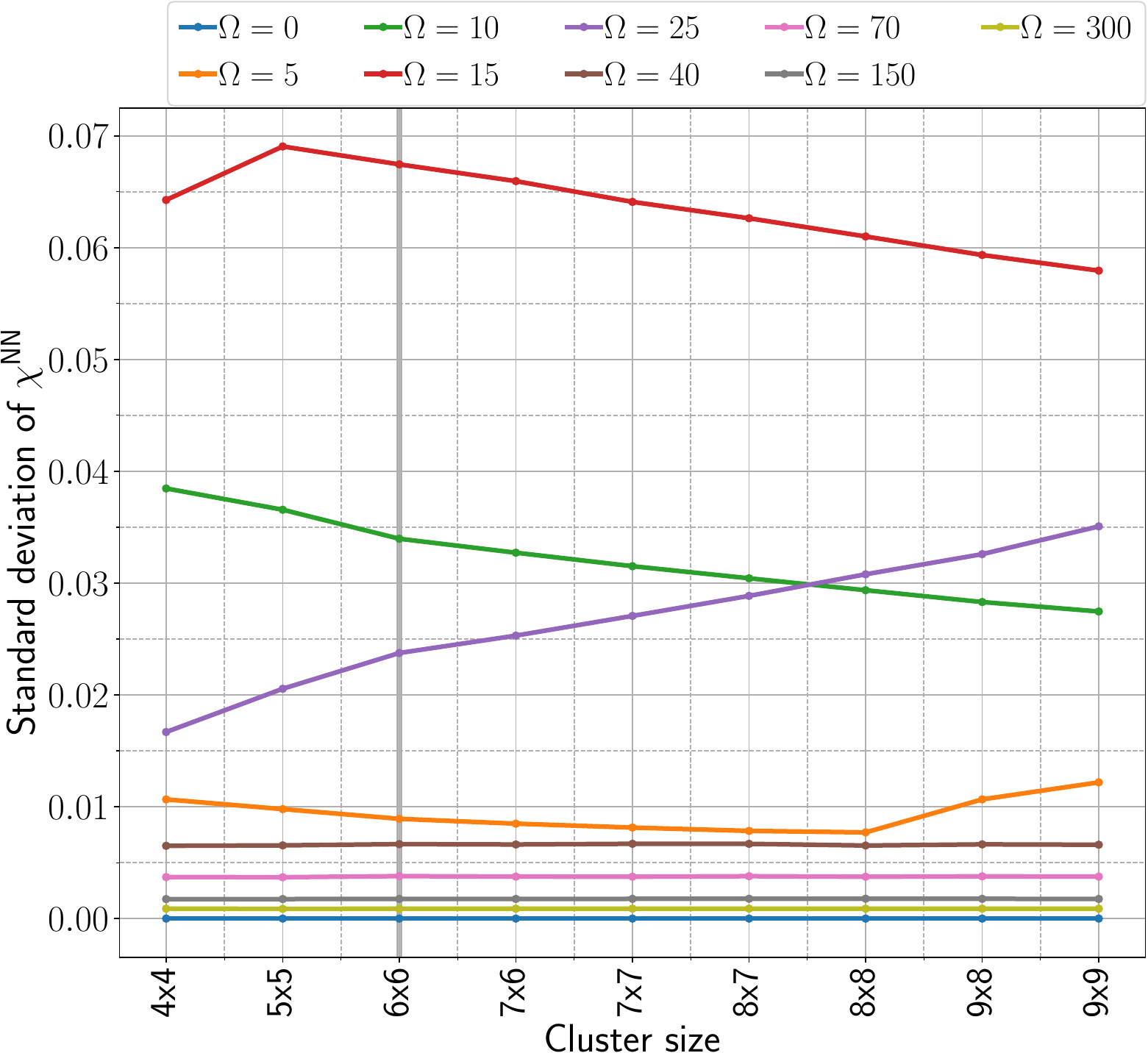}
      \caption{\textbf{Standard deviation of the $\chi^{\text{NN}}$} as a function of cluster size for a set of transverse fields $\Omega$. The grey vertical bar pins down the size above which the GNN extrapolates. The largest variances, across all cluster sizes, are obtained when $J_{\langle i,j\rangle} \approx \Omega$. $\Omega$ is shown in units of $\text{rad}\cdot \mu s^{-1}$.}
  \label{fig:NN_corr_variance_delta}
\end{figure}

To better understand why some values of $\Omega$ perform better than others in Fig.~\ref{fig:metrics_deltas}, we plot in Fig.~\ref{fig:NN_corr_variance_delta} the standard deviation (STD) of $\chi^{\text{NN}}$, denoted $\text{STD}[\chi^{\text{NN}}]$, as a function of the cluster sizes for all $\Omega$'s considered. In the range $\Omega \simeq \left[10,25\right] \ \text{rad}\cdot\mu s^{-1}$, fluctuations in the correlation function surpass that of all other $\Omega$'s by an order of magnitude at least, signaling the magnetic phase transition in the vicinity. Moreover, in that regime, the trend of $\text{STD}[\chi^{\text{NN}}]$ is far from being constant with the cluster size. Indeed, spin interactions in the critical regime become more long-range, which renders the training more size-dependent and greatly reduces extrapolation performance at larger cluster sizes -- physics is different between cluster sizes. Even for $\Omega = 25 \ \text{rad}\cdot\mu s^{-1}$, the fluctuations increase for all values of cluster sizes, in line with the critical behavior of correlation functions when passing a quantum phase critical point. When the quantum critical point is passed, but the standard variation of $\chi^{\text{NN}}$ remains roughly within $1\times10^{-2}\gtrsim\text{STD}[\chi^{\text{NN}}]\gtrsim 1\times10^{-3}$, this is where we see an optimal gain in performance. Hence, when $\Omega > J_{\langle i,j\rangle}$, but not too much, one finds a sweet spot. However, one could generally take a $\Omega$-history for $\Omega > J_{\langle i,j\rangle}$. The values of $\Omega = \{0, 5, 300\} \ \text{rad}\cdot\mu s^{-1}$ all feature a $\text{STD}[\chi^{\text{NN}}]$ that falls below $1\times10^{-3}$, too small for the GNN to pick up useful signals.

\subsubsection{Additional studies}
\label{sec:additional_studies}

We also study the GNN extrapolation dependence on the physical system sizes included in the training dataset. In Appendix~\ref{app:amount_training_datasets_training_datasets}, we show that unsurprisingly the GNN trained on a training dataset including a larger variety of system sizes performs and extrapolates better than when trained on smaller physical systems. Surprisingly, the smaller training datasets diminish the GNN performance the strongest in the regime of small physical systems, indicating the difficulty of learning the edge effects of physical systems. Moreover, in Appendix~\ref{app:target_optimization} we show that it is possible to predict also NNN distances between atoms, but it requires larger physical systems included in the training set to keep the same GNN performance. This hints at the existence of a critical physical system size that needs to be included in the dataset for the GNN to extrapolate properly, and its dependence on the locality of the Hamiltonian.

\subsection{Predictions from snapshots}
\label{sec:predictions_snapshots}

In this section, we ask how the GNN performs when the input correlation functions are estimated from a collection of successive projective measurements of the wave function, \textit{i.e.}, trained from snapshot training datasets instead of exact training datasets. Note that all the assertions made in the previous sections carry over to this training since the GNN does not change, nor does the physics. Here, we still train the GNN with a training dataset whose samples are made up of correlation functions (local and nonlocal) estimated from a number of snapshots (projective measurements of the wave function), nicknamed snapshot data. One could naively think that predicting the NN relative distances of snapshot data with the GNNs trained with exact correlators (directly measured in the ground state) would lead to a better outcome as opposed to the case where the GNNs are trained with snapshot data. We tried predicting the relative distances from snapshot data using a GNN model trained with exact correlators, but it produced worse results compared to a GNN model trained on the snapshot data (not shown). In fact, when measuring directly from the wave function, no errors coming from the finite number of projective measurements spoil the training dataset, and the GNN can thereof dedicate all its variational parameters to `learn' the bijection presented in Sec.~\ref{app:bijection_proof}. In Sec.~\ref{sec:comparison_best_training_scenarios}, we show how the GNN performance depends on the number of snapshots used to estimate correlation functions. In Sec.~\ref{sec:comparison_GNN_MLP}, we compare the performance of the GNN and the MLP on the snapshot data to determine the advantage of the graph preprocessing. Finally, in \seclab~\ref{sec:snapshot_size_effect}, we study whether the GNN error on snapshot data scales according to the central limit theorem.

\begin{figure}[h!]
  \centering
    \includegraphics[width=0.95\columnwidth]{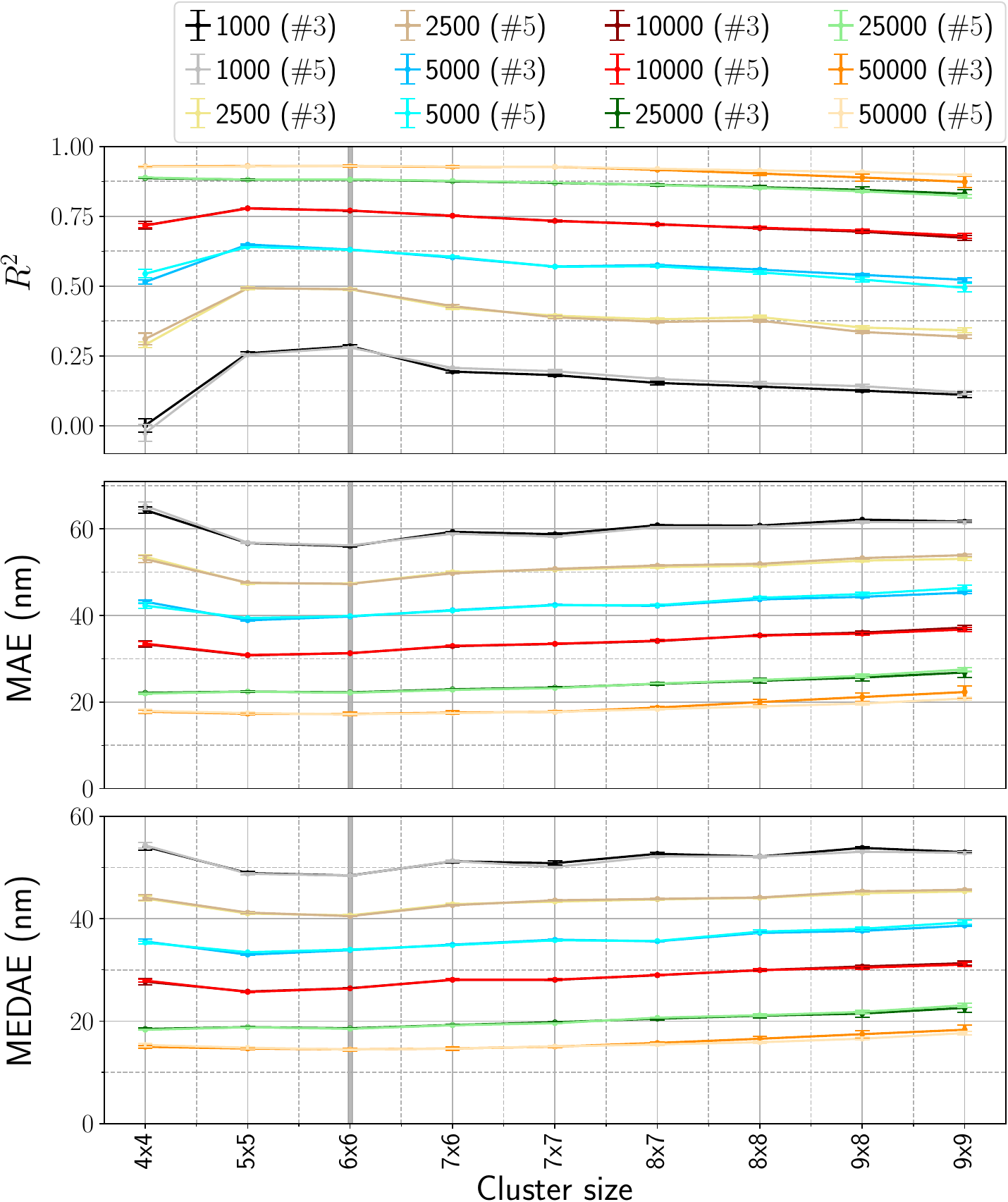}
      \caption{\textbf{Performance of training cases $\#3$ and $\#5$ using snapshot data.} A view of the metrics so far considered as a function of the cluster size. The two best cases, namely cases $\#3$ and $\#5$, have been retained. The dataset has been generated using a finite number of snapshots to estimate the observables, namely $N_{\text{smpl}}=1000, 2500, 5000, 10000, 25000$, and $50000$ snapshots. The error bars show the standard error on five distinct and independent trainings. The vertical gray lines indicate the point below which cluster sizes are considered in the snapshot training dataset.}
  \label{fig:metrics_SNPT}
\end{figure}

\subsubsection{Comparison between best training scenarios}
\label{sec:comparison_best_training_scenarios}

The two best cases computed using Z-observable measurements, presented in Sec.~\ref{sec:results}, are compared against each other in Fig.~\ref{fig:metrics_SNPT}, namely the cases $\#3$ and $\#5$ (see Table~\ref{table:framed}). Furthermore, since when dealing with snapshots to estimate correlation functions like in experiments one should be able to carry out measurements in the X-basis as well (see Eq.~\eqref{eq:NN_spin_corr_X}), we check the GNN performance when combining Z- and X-correlators in Fig.~\ref{fig:metrics_SNPT_ZX}. The colors were carefully chosen in Figs.~\ref{fig:metrics_SNPT} and \ref{fig:metrics_SNPT_ZX} to match the snapshot sample size.

Looking first at Fig.~\ref{fig:metrics_SNPT}, overall, one notices that the case $\#5$ slightly underperforms case $\#3$, especially for smaller $N_{\text{smpl}}$'s. The architectural benefits seem to be washed out when training with snapshot data, as is clearly apparent when looking back at the case $\#5$ in Fig.~\ref{fig:metrics_training_GNN}. Evidently, increasing the number of snapshots is very beneficial to the training. When comparing Fig.~\ref{fig:metrics_training_GNN} to Fig.~\ref{fig:metrics_SNPT}, having in mind the bijection proof detailed in Appendix~\ref{app:bijection_proof}, it seems obvious the almost-perfect rendering of the GNN (see cases $\#3$, $\#5$ and $\#6$ in Fig.~\ref{fig:metrics_training_GNN}) is unlocked due to the exact and precise correlators provided, and that the main limitation seen in Figs.~\ref{fig:metrics_SNPT} and \ref{fig:metrics_SNPT_ZX} relates to the statistically diluted physical content in snapshots. 

\begin{figure}[t]
  \centering
    \includegraphics[width=0.95\columnwidth]{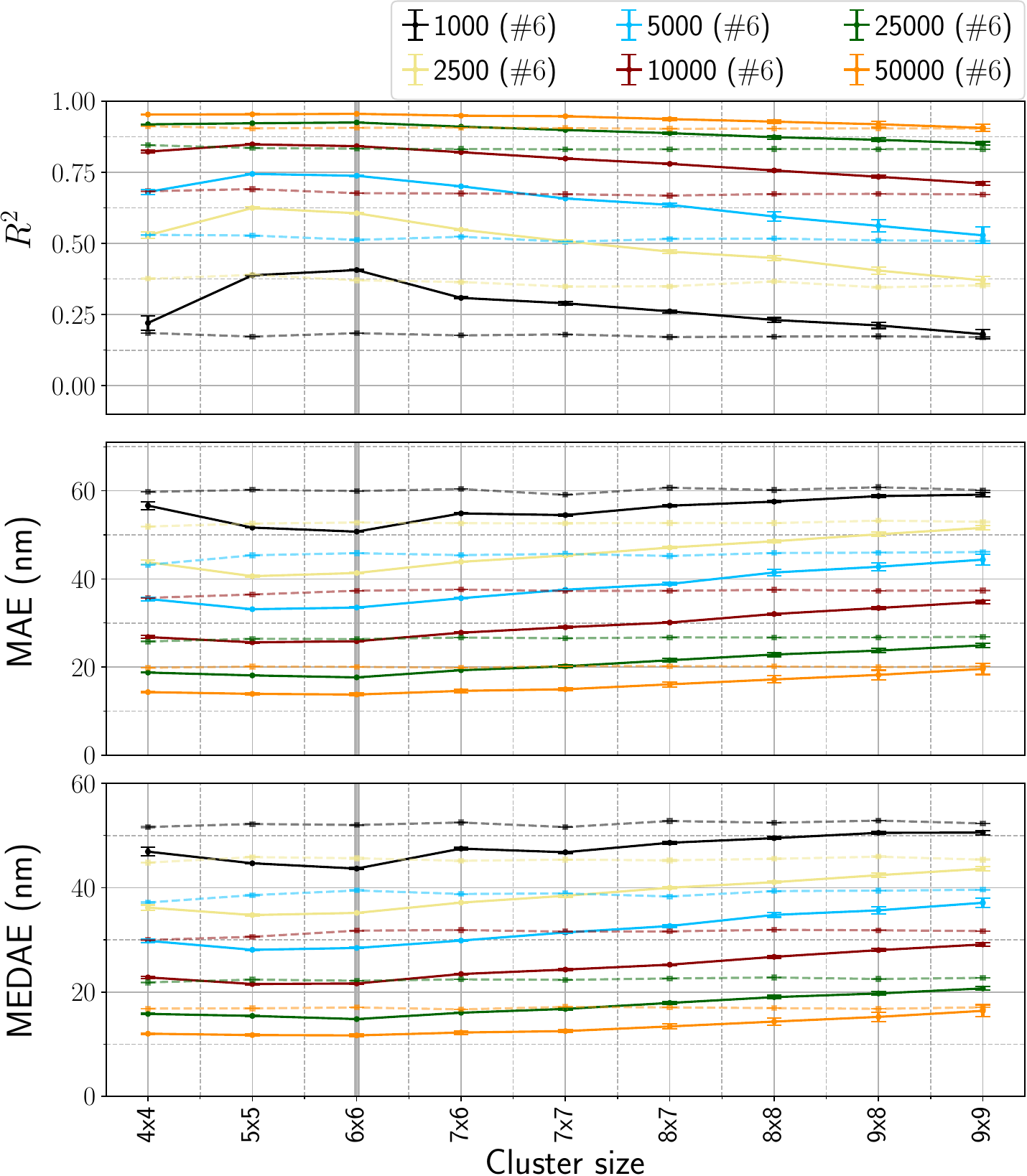}
      \caption{\textbf{Performance of the training case $\#6$ using snapshot data.} A view of the metrics so far considered as a function of the cluster size. The case $\#6$ is considered. The dataset has been generated using the same finite number of snapshots as in Fig.~\ref{fig:metrics_SNPT}. The error bars show the standard error on five distinct and independent trainings. Again, the gray vertical line indicates the cluster size threshold above which neural network models extrapolate. The dotted curves with corresponding colors are predictions from an MLP (including error bars).}
  \label{fig:metrics_SNPT_ZX}
\end{figure}

To remedy the poorer information content in a snapshot dataset, we show in Fig.~\ref{fig:metrics_SNPT_ZX} how the GNN performs if we include both X- and Z-observable measurements. Just like observed in Fig.~\ref{fig:metrics_training_GNN}, we get a substantial improvement on the metrics considered, especially the MAE and MEDAE, when the $\chi^{\text{NN}}$ and $\chi^{\text{NNN}}$ correlators are computed with the Z and X Pauli strings. Doing so, the MAE and MEDAE hover around a discrepancy of $10\%$ on snapshot data utilizing 10000 projective measurements. However, the combination of the Z- and X-observables turns out not to be as beneficial as the doubling of $N_{\text{smpl}}$.

\subsubsection{Comparison between the GNN and the MLP performance}
\label{sec:comparison_GNN_MLP}

In Fig.~\ref{fig:metrics_SNPT_ZX}, we compare the predictive errors and scaling of the GNN (full lines) and a simple MLP of the same architecture as the GNN's task network (dashed lines). The same MLP is trained to predict an NN distance between a pair of atoms from the NN spin-spin correlation for all the atoms, disorder realizations, and cluster sizes included in the training set.
The resulting MLP has a low predictive error, which is likely due to the bijection proved in Appendix~\ref{app:bijection_proof}.
Moreover, the MLP performance is almost constant across cluster sizes and is comparable to the performance of the GNN at the largest cluster sizes.

At the same time, the GNN outperforms the MLP over the range of cluster sizes considered in this study. This gain shows the advantage of graph preprocessing, which can lead to learning the edge effects and more accurate predictions for cluster sizes, where edge effects play a role.

\begin{figure}[t]
  \centering
    \includegraphics[width=0.95\columnwidth]{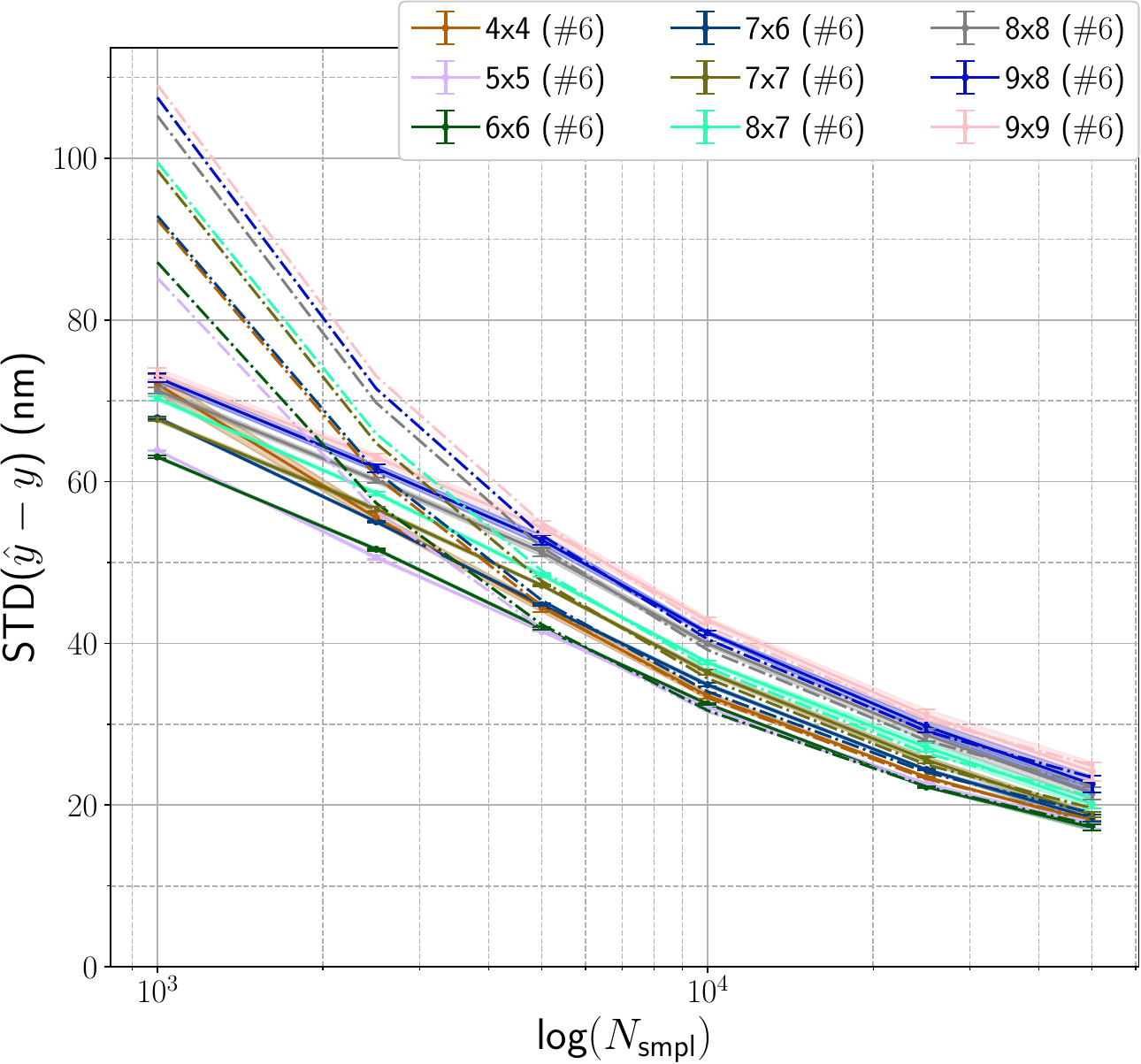}
      \caption{\textbf{Effect of the snapshot sample size $N_{\text{smpl}}$.} Standard deviations (STDs) of the difference between NN relative distance predictions and targets for case $\#6$ using snapshots (solid curves) for all cluster sizes considered. The corresponding dotted curves whose color matches that of the solid curve show how the STD would evolve with the sample size according to the central limit theorem, taking the snapshot data with $N_{\text{smpl}}=\{5000,10000,25000,50000\}$ as reference. The error bars show the standard error on five distinct and independent trainings of the PNA GNN.}
  \label{fig:STDS_SNPT_ZX}
\end{figure}

\subsubsection{Snapshot size effect}
\label{sec:snapshot_size_effect}

To check if the GNN performance on snapshot data correlates with the dataset size according to the central limit theorem, we plot in Fig.~\ref{fig:STDS_SNPT_ZX} the standard deviation of the difference between the network predictions and targets for various snapshot sample sizes ($N_{\text{smpl}}$) for all cluster sizes (solid curves). Additionally, we use the snapshot data with $N_{\text{smpl}}=\{5000,10000,25000,50000\}$ to fit the power-law scaling $\propto 1/\sqrt{N_{\text{smpl}}}$ of the central limit theorem (dotted curves with corresponding colors). Since the validity of the central limit theorem holds for large sample sizes, the profile of the standard deviation of the snapshot data with $N_{\text{smpl}}=\{1000,2500\}$ deviates from the central limit theorem. Without much surprise, the standard error between the predictions and truths diminishes with the number of snapshots $N_{\text{smpl}}$ used to calculate the observables. For smaller values of $N_{\text{smpl}}$, the discrepancy between the solid and dotted curves signals that the GNN is trying to learn something from the statistical uncertainty arising from snapshot data, leading to noise overfitting.

Estimating large covariance matrices from finite samples is an area that received a lot of attention in the last few decades~\cite{lam2020high}. Truncating the eigenvalue decomposition of the empirical covariance matrix (connected correlators), based on random matrix theory, is one obvious procedure that has been in use in applications from image denoising to portfolio construction~\cite{sengupta1999distributions,laloux1999noise,laloux2000random}. More recently, further progress has been made on shrinkage-based methods~\cite{ledoit2022power}. Furthermore, one could draw snapshots from an MPS state in random computational bases, \textit{i.e.} randomly applying a unitary rotation on each site before doing a measurement in the Z-basis, since this can help the estimation of observables to convergence~\cite{stoudenmire2010mpssampling}. In the future, we will investigate whether these methods improve our network's performance on snapshot data. 

\section{Conclusions and discussion}\label{sec:conclusions}

In this work, we have made use of the Principal Neighbor Aggregation (PNA) algorithm to predict the positions of Rydberg atomic arrays in 2D cold atom arrays. The training dataset was produced by simulating the ground state of a transverse-field Ising model using the Density Matrix Renormalization Group (DMRG) algorithm. Overall, the graph neural network (GNN) employed demonstrated a clear capacity to extrapolate outside its training domain -- also the multi-layer perceptron (MLP) considered at the end of the work --, be it in terms of system sizes or shape. The GNN and MLP scalability is probably enabled from the quantum system's side by the effective locality of the studied Hamiltonian \cite{buhrman2024beatinggroversearch} (decay of the spin coupling as $1/R^6_{i,j}$). The training input data was the local magnetization, the nearest-neighbor (NN) and next-nearest-neighbor (NNN) correlation functions, and the targets were NN and NNN relative distances between atoms in the array, depending on the training case scenarios (see Table~\ref{table:framed}). To narrow down the best training conditions for the GNN, various training procedures were followed, each with its own efficiency. The correlators were all computed in the Z measurement basis and in some cases Z- and X-correlators were used.

Depending on the number of experimental snapshots that can be generated, if one does not have knowledge of the best transverse field $\Omega$ to train the GNN with, it turns out that generating samples with the same array disorder across different values of $\Omega$ helps tremendously. The values of $\Omega$ making up this array should be selected roughly above, but not too far, from $\Omega\simeq C_6/\bar{R}^6_{\langle i,j\rangle}$, where $C_6$ represents the van der Waals interaction and $\bar{R}_{\langle i,j\rangle}$ represents the nominal NN distance between atoms. Alternatively, active learning could be employed to identify optimal $\Omega$'s to train the GNN~\cite{Dutt2023activeHlearn}.

The GNN performance depends on the physical data it utilized during the training. In high-quality numerical/experimental setups, only the NN spin correlation function, measured in one single measurement basis, should suffice to predict in a one-to-one mapping the NN relative distances between the Rydberg atoms. This is what is seen when exact correlators are used (directly computed from the full wave function). However, when using snapshot data, uncertainties occur and significantly reduce the performance of the GNN. To reduce this unwanted spoiling of the data, a combination of both Z- and X-correlators is used. In doing so, to get average predictions over the NN relative distances well below 10$\%$ discrepancy, more than 10000 projective measurements need to be carried out to estimate the spin-spin correlators. However, increasing the number of snapshots is essential to improve the performance. We also showed, by comparing the GNN performance to that of the MLP, that the graph preprocessing done by the GNN layer allows the inclusion of edge effects and provides more accurate predictions for cluster sizes, where edge effects play a role.

On the technical side, some architectural properties of the GNN can be exploited to boost the model's performance. For instance, enabling NNN edges in input graphs without assigning physical values to their features is of great help. This way, `skip connections' are created through which the neural network can back-propagate and help the flow of information get through the neural layers. The local magnetization provides little information to the training and should also be set to identity. Including more cluster sizes in the training datasets significantly improves the extrapolation capabilities of the GNN, although three small cluster sizes seem to suffice to predict much larger cluster sizes of atoms with very high fidelity (with the exact correlators).

Based on this systematic study of the PNA GNN, we can easily identify various avenues to fine-tune the neural network and tailor it to specific needs. The points brought up in what follows should further enhance the GNN performance in predicting the NN relative distances. As a first point, \textbf{i}) one could add more cluster sizes to the training dataset $\{4\times 4,5\times 5,6\times 6\}$ considered here, such as sizes $4\times 5$, $5\times 6$, etc. Secondly, \textbf{ii}) one could increase the number of correlators to include longer-range ones, thereby amplifying the over-completion of the task of inferring NN relative distances. Thirdly, \textbf{iii}) one could enable higher-order graph edges that are not associated with any correlator in the datasets since this has proven to be beneficial (cases $\#4$ and $\#5$) to the performance; these `weightless' edges would act like `skip connections' in the back-propagation of the gradients during the GNN optimization. Lastly, \textbf{iv}) the $\vec{\Omega}$-history used to produce the training dataset could be more carefully chosen to remain near a quantum phase transition with an increased number of $\Omega$ slices. 

We are currently exploring applications of the method described in this article to time-dependent experimental data on Rydberg atomic arrays. Looking forward, the GNN (and MLP) could be trained on various types of Hamiltonians~\cite{nandy2024reconstructing} and, therefore, predict many possible Hamiltonian terms and parameters, e.g., identifying unknown sources of noise from the measurement data instead of only predicting parameters of a pre-defined Hamiltonian. To achieve a larger generality of the GNN architecture described in this work, one needs to include many-body edges in the GNN layers and task networks whose input are multiple nodes. It is also interesting to consider the limitations of the GNN scalability to larger system sizes and whether it depends on the Hamiltonian locality (power-scaling of the interaction decay). Eventually, one could envision feedback control of the experimental setup to reduce uncertainties on the simulated Hamiltonian.

\begin{acknowledgments}
We are grateful to Antoine Browaeys, Thierry Lahaye, and their group at IOGS for getting us interested in this problem 
and for useful discussions. We thank Krzysztof Maziarz for suggesting Principal Neighborhood Aggregation as the GNN architecture. We also thank Lucy Reading-Ikkanda for the graphical design of Figs.~\ref{fig:intro} and~\ref{fig:GNN}. O.S. thanks Sékou-Oumar Kaba for useful discussions regarding the GNN implementations. A.D. thanks Liubov Markovich for the discussions on Hamiltonian learning techniques. O.S. acknowledges support from the Swiss National Science Foundation through the Postdoc.Mobility fellowship. A.D. acknowledges support from the Dutch National Growth Fund (NGF), as part of the Quantum Delta NL programme. This work was granted access to the HPC resources of TGCC and IDRIS under the allocation A0170510609 attributed by GENCI (Grand Equipement National de Calcul Intensif). It has also used high performance computing resources of IDCS (Infrastructure, Données, Calcul Scientifique) under the allocation CPHT 2024. This work is part of HQI initiative (\url{www.hqi.fr}) and is supported by France 2030 under the French National Research Agency award number ANR-22-PNCQ-0002.
The Flatiron Institute is a division of the Simons Foundation.
\end{acknowledgments}

\section*{Data availability}

The DMRG datasets presented in this work is partly available on Zenodo via this DOI \texttt{10.5281/zenodo.15525077}, along with all the plotting scripts and accompanying data to reproduce the figures. The DMRG code used to produce the samples can be provided upon request to the authors. The instructions to reproduce the figures and train the GNN with the provided DMRG dataset are available through this \href{https://github.com/oliviersimard/Hamiltonian_Learning_GNN.git}{link}.

\section*{Code availability}

The DMRG code used to generate the datasets is available upon request to the authors and was written using the library ITensors.jl~\cite{10.21468/SciPostPhysCodeb.4}. The PNA GNN codes to train and test upon the datasets are available on \url{https://github.com/oliviersimard/Hamiltonian_Learning_GNN.git}. The Python scripts used to produce the figures are also available on Zenodo via this DOI available on Zenodo via this DOI \texttt{10.5281/zenodo.15525077}. The GNN code was written using utilities from PyTorch Geometric~\cite{paszke_pytorch_2019}.

\bibliographystyle{apsrev4-1_our_style}
\bibliography{BIB_ML-in-Paris, additional-ref}

\clearpage

\appendix

\section{Proof of bijection between interactions and spin-spin correlators}
\label{app:bijection_proof}

In this section of the Appendix, we prove a general result that provides a theoretical foundation for the Hamiltonian reconstruction approach implemented in this work. In part, we show that it suffices to measure NN spin correlation functions to determine uniquely the NN relative distances between atoms and, therefore, the full set of interactions. We show that there is, in general, a bijective correspondence relating the matrix of spatially dependent correlation functions $c_{i,j}$ to the matrix of interactions $J_{i,j}$ between Ising spins in the TFIM. This result can be viewed as a generalization of the Hohenberg-Kohn (HK) theorem of density functional 
theory~\cite{HohenbergKohn,Garrigue2019HohenbergKohn}, which establishes a bijection between the local one-body potential and the local density of a system of interacting quantum particles, and as a quantum many-body generalization of the Henderson theorem~\cite{HENDERSON1974197}. We note that previous work has 
demonstrated that a neural network can indeed be trained to learn the HK correspondence in interacting fermion models~\cite{RobledoMorenoHK}. 

We write the TFIM Hamiltonian in the form:
\begin{equation}\label{eq:TFIM_bijection}
\ham = \Omega \sum_i \hat{\sigma}_i^x + \delta \sum_i \hat{\sigma}_i^z + \sum_{i<j} J_{i,j} \hat{\sigma}_i^z \hat{\sigma}_j^z
\end{equation}
where $i<j$ indicates that each link is counted only once. We define the correlation functions as: 
$c_{i,j} = \langle\hat{\mathbf{S}}_i^z \hat{\mathbf{S}}_j^z\rangle$, where $\langle \cdot\rangle$ indicates the expectation value and $\hat{\mathbf{S}}_i^z$ is the spin operator in the Z-basis acting on the $i$-th site, related to the Z Pauli matrix via $\hat{\mathbf{S}}_i^z = \frac{\hbar}{2} \mathbf{\hat{\sigma}}_i^z$. We prove that, for fixed $\Omega\neq 0$ and $\delta$, there is a unique set of interactions $J_{i,j}$ that yields a set $c_{i,j}$, and conversely: $J_{i,j} \Leftrightarrow c_{i,j}$. 
The condition $\Omega\neq 0$ is important since, for example, when $\Omega=\delta=0$, for all choices of $\{J_{i,j}\}$ with each individual $J_{i,j}<0$ the ground-state is the 
fully polarized ferromagnet and the $c_{i,j}$'s are obviously the same, irrespective of the specific values of the interactions. 

First, we need two properties of the ground states of the Hamiltonian when $\Omega\neq 0$. We use  $\ket{\sigma} \equiv \ket{\sigma_1,\cdots,\sigma_N}$ to denote the eigenstates of all the $\hat{\mathbf{S}}^z_i$'s. These properties are consequences of the stoquasticity \cite{bravyi2008complexity} of the TFIM, and the irreducibility \cite{crosson2017quantum} of the canonical ensemble density matrix in the standard basis. We give an explicit proof for completeness.
\begin{lemma}
\label{lemma:NonZeroOmega}
For the TFIM Hamiltonian $\ham$ defined by equation \ref{eq:TFIM_bijection} with $\Omega\neq 0$, we have the following properties:
\begin{enumerate}
  \item the ground state $\ket{\Psi}$ of $\ham$ is unique,
  \item $\langle\sigma|\Psi\rangle\neq 0$ for all $\ket{\sigma}$.
\end{enumerate}
\end{lemma}
\begin{proof}
First, it is enough to prove the result for $\Omega<0$, since we can find a unitary transformation $U:=\prod_j (i\hat\sigma_j^y)$ that changes the sign of $\Omega$ and $\delta$ but leaves $\{J_{ij}\}$ intact. Since each $i\hat\sigma_j^y$ locally maps $\ket{\uparrow}$ to $-\ket{\downarrow}$ and $\ket{\downarrow}$ to $\ket{\uparrow}$ at the $j$-th site, the $\ket{\sigma}$'s are permuted among each other modulo an extra sign. Hence, the statement proved for the transformed Hamiltonian is true for the original sign.

Second, it is enough to prove that $V=e^{-\beta \ham}, \beta>0,$ is a positive matrix, namely it has all positive entries. Then, the Perron-Frobenius theorem \cite{perron1907theorie,frobenius1908matrizen,frobenius1912matrizen} tells us that the dominant (Perron) eigenvalue is simple (nondegenerate) and the corresponding (Perron) eigenvector has all positive components. Note that $V$ is Hermitian with positive eigenvalues with the Perron eigenvalue being $\exp(-\beta E)$ when $E$ is the ground-state energy of $\ham$ and the Perron eigenvector is the ground state $\ket{\Psi}$.
Now we break up $\ham=\ham_0+\Omega\sum_i\hat\sigma_i^x$ where
\begin{equation}\label{eq:H_0}
\ham_0 =\delta \sum_i \hat{\sigma}_i^z + \sum_{i<j} J_{i,j} \hat{\sigma}_i^z \hat{\sigma}_j^z
\end{equation}
We are ready to use the series expansion in the interaction picture (see e.g. Ref.~\cite{sandvik2019stochastic}) for $V$ entries:
\begin{align}
\langle\sigma'|V|\sigma\rangle = &\langle\sigma'|e^{-\tau \ham}|\sigma\rangle \nonumber\\
=
\sum_{n=0}^\infty(-\Omega)^n&\sum_{i_1,i_2\cdots,i_n}\int_{\tau_{n-1}}^\beta d\tau_{n}\cdots\int_{\tau_1}^\beta d\tau_2\int_{0}^\beta d\tau_{n}\nonumber\\
&\langle\sigma'|e^{-\beta \ham_0}\hat\sigma_{i_n}^x(\tau_n)\cdots\hat\sigma_{i_2}^x(\tau_2)\hat\sigma_{i_1}^x(\tau_1)|\sigma\rangle
\label{eq:SSE}
\end{align}
where $\hat\sigma_{i}^x(\tau):=e^{\tau\ham_0}\hat\sigma_{i}^xe^{-\tau\ham_0}$. This operator causes a spin flip at site $i$. This convergent series has every term that is nonnegative. Moreover between any $\ket{\sigma}$ and $\ket{\sigma^{\prime}}$, there is at least one term with at most $N$ $\hat\sigma_{i_k}^x$'s that is nonzero. Hence $V$ has all positive entries.
\end{proof}
Now we come to the main theorem, adapted to the TFIM, which assumes the $J$-representability of the connected correlation functions~\cite{PhysRevA.26.1200,ENGLISCH1983253}.

\begin{theorem}[Hohenberg-Kohn-Henderson Theorem for the TFIM]
For the TFIM as defined in Eq.~\eqref{eq:TFIM_bijection}, for fixed nonzero $\Omega$ and fixed $\delta$, there exists a bijection between the $J$-representable correlation functions 
$c_{i,j} = \langle\hat{\mathbf{S}}_i^z \hat{\mathbf{S}}_j^z\rangle$ and the set of interactions $J_{i,j}$, 
$J_{i,j} \Leftrightarrow c_{i,j}$.
\end{theorem}
\begin{proof}
The direction $J_{i,j}\implies c_{i,j}$ is obvious. We now prove that $c_{i,j}\implies J_{i,j}$ by {\it reductio ad absurdum}, following the same logic for the classic HK theorem~\cite{HohenbergKohn}. By the aforementioned lemma, the ground state of $\ham$ is nondegenerate.

For fixed values of $\Omega \ (\neq 0)$ and $\delta$, we therefore start by assuming that one can construct two ground-state wave functions, $\Psi_1$ and $\Psi_2$, with identical $c_{i,j}$'s but different $J_{i,j}$'s, \textit{i.e.} $J_{i,j}^{(1)}$ and $J_{i,j}^{(2)}$. Therefore, the Hamiltonians $\ham_1$ and $\ham_2$, whose ground states are $\Psi_1$ and $\Psi_2$, with ground-state energies $E_1$ and $E_2$, 
are identical besides the $J_{i,j}$ term:
\begin{gather}
    \ham_{1} = \Omega \sum_i \hat{\sigma}_i^x + \delta \sum_i \hat{\sigma}_i^z + \sum_{i<j} J_{i,j}^{(1)} \hat{\sigma}_i^z \hat{\sigma}_j^z,\\
    \ham_{2} = \Omega \sum_i \hat{\sigma}_i^x + \delta \sum_i \hat{\sigma}_i^z + \sum_{i<j} J_{i,j}^{(2)} \hat{\sigma}_i^z \hat{\sigma}_j^z.
\end{gather}
We first show that $\Psi_1\neq \Psi_2$. Assuming that instead $\Psi_1=\Psi_2\equiv\Psi$, one 
then has: $(\ham_1-\ham_2)|\Psi\rangle = (E_1-E_2)|\Psi\rangle$. Expanding $|\Psi\rangle$ on a basis of 
eigenstates $|\sigma\rangle \equiv |\sigma_1,\cdots,\sigma_N\rangle$ of the $\hat{\mathbf{S}}^z_i$'s, one then obtains:
\begin{equation}
    \sum_{i<j} (J_{i,j}^{(1)}-J_{i,j}^{(2)}) \sigma_i \sigma_j  = 
    (E_1-E_2) 
\label{eq:SamePsiContradiction}    
\end{equation}
which must be satisfied for {\it all} configurations $\sigma$ for which $\langle\sigma|\Psi\rangle\neq 0$. But $\langle\sigma|\Psi\rangle\neq 0$ is true for all possible $\sigma$ by the previous lemma. So Eq.~\eqref{eq:SamePsiContradiction} is satisfied for all $2^N$ configurations of $\sigma$. If we sum this equation over all possible configurations, we get $0=2^N(E_1-E_2)$, as $\sum_\sigma\sigma_i\sigma_j=0$ for any $i\neq j$. Thus, $E_1=E_2$.

Now we have the equation
\begin{equation}
    \sum_{i<j} (J_{i,j}^{(1)}-J_{i,j}^{(2)}) \sigma_i \sigma_j  = 0.
\label{eq:SamePsiContradictionCotd} 
\end{equation}
To arrive at a contradiction, we will show that $J_{i,j}^{(1)}=J_{i,j}^{(2)}$ for every $i,j$ pair with $i<j$. After picking one such specific ${i,j}$ pair, choose all configurations where $\sigma_i=1, \sigma_j=1$. Call this set $s_{ij}$.
When $i'<j'$, $\sum_{\sigma\in s_{ij}}\sigma_{i'}\sigma_{j'}=2^{N-2}\delta_{ii'}\delta_{jj'}$, $\delta_{ij}$ being the Kroenecker delta. Summing Eq.~\eqref{eq:SamePsiContradictionCotd} over $\sigma\in s_{ij}$, we get $2^{N-2}(J_{i,j}^{(1)}-J_{i,j}^{(2)})=0$. This is in contradiction to the couplings being distinct.
Hence, $\Psi_1\neq \Psi_2$. 

Having established that, we can now use the variational principle in the form of a 
{\it strict} inequality for both Hamiltonians:
\begin{gather}
    E_1 = \expval{\ham_1}{\Psi_1} < \expval{\ham_1}{\Psi_2}, \label{eq:var_principle_1}\\
    E_2 = \expval{\ham_2}{\Psi_2} < \expval{\ham_2}{\Psi_1} \label{eq:var_principle_2}.
\end{gather}
Let us expand the right side of the inequalities:
\begin{align*}
    \expval{\ham_1}{\Psi_2} &= \expval{(\ham_1 - \ham_2)}{\Psi_2} + \expval{\ham_2}{\Psi_2}\\ 
    &= \expval{\sum_{i<j} (J_{i,j}^{(1)} - J_{i,j}^{(2)}) \hat{\sigma}_i^z \hat{\sigma}_j^z}{\Psi_2} + E_2\\
    &= \sum_{i<j} (J_{i,j}^{(1)} - J_{i,j}^{(2)}) c_{i,j} + E_2,\\
    \expval{\ham_2}{\Psi_1} &= \expval{\sum_{i<j} (J_{i,j}^{(2)} - J_{i,j}^{(1)}) \hat{\sigma}_i^z \hat{\sigma}_j^z}{\Psi_1} + E_1\\
    &= - \sum_{i<j} (J_{i,j}^{(1)} - J_{i,j}^{(2)}) c_{i,j} + E_1.
\end{align*}
Now let us add Eqs.~\eqref{eq:var_principle_1} and~\eqref{eq:var_principle_2}:
%
\begin{equation}
E_1 + E_2 < E_1 + E_2 + \sum_{i<j} (J_{i,j}^{(1)} - J_{i,j}^{(2)}) c_{i,j} - \sum_{i<j} (J_{i,j}^{(1)} - J_{i,j}^{(2)}) c_{i,j}
\end{equation}
and hence: 
\begin{equation}
    E_1+E_2 < E_1+E_2
\end{equation}
is a clear contradiction. 

We thus conclude by \textit{reductio ad absurdum} that one cannot have two ground-state wave functions with identical $c_{i,j} = \langle \hat{\mathbf{S}}_i^z \hat{\mathbf{S}}_j^z\rangle$ and different $J_{i,j}$, and hence that  $\langle \hat{\mathbf{S}}_i^z \hat{\mathbf{S}}_j^z\rangle \implies J_{i,j}$.

\end{proof}

\section{Number of physical system sizes in training datasets}
\label{app:amount_training_datasets_training_datasets}

In this section of the Appendix, the effect of the number of physical system sizes included in the training datasets is studied. Here, we utilize exact training datasets.

In Fig.~\ref{fig:metrics_training_sizes}, for the set of sizes $\{4\times 4,5\times 5\}$ and $\{4\times 4,5\times 5,6\times 6\}$, we plot the coefficient of determination $R^2$~\eqref{eq:equation_coefficient_determination} in the top panel, the MAE~\eqref{eq:mean_absolute_error} in the middle panel, and the median of the MAE (MEDAE) in the bottom panel. The case $\#3$ is utilized. Without much surprise, the GNN performs much better at extrapolating relative NN distances $\Delta R_{\langle i,j\rangle}$ for clusters whose sizes lie outside the training set because of system size effects. It may, however, come as a surprise that the smaller system sizes, \textit{i.e.} $4\times 4$ and $5\times 5$, perform worse than that $6\times 6$ within the training size domain, but it seems that this effect mostly comes from the inclusion of the $\chi^{\text{NNN}}$, as seen in Fig.~\ref{fig:metrics_training_GNN}. Since smaller clusters are subparts of larger clusters, the information of the smaller clusters aggregates and makes it possible to predict with high precision the $\Delta R$'s of clusters with sizes beyond the training set.

\begin{figure}[h!]
  \centering
    \includegraphics[width=0.95\columnwidth]{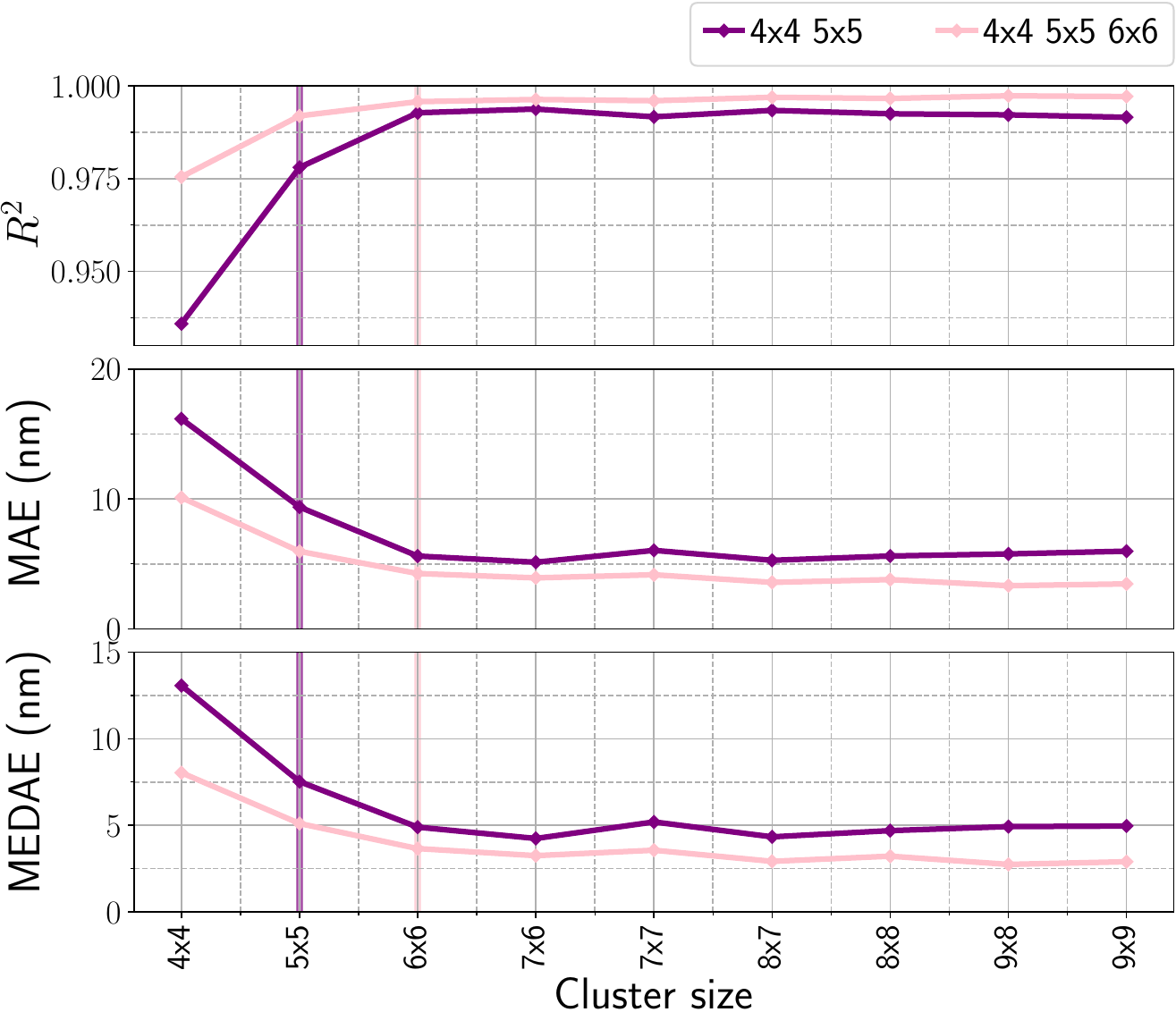}
      \caption{\textbf{Training dataset cluster size effect.} Incidence of the cluster sizes included in the exact training dataset on various metrics. In violet, clusters of sizes $4\times 4$ and $5\times 5$ were considered, whereas in light pink, one extra cluster size $6\times 6$ was considered. The vertical lines refer to the point beyond which sizes lie outside of the training range, \textit{i.e.} these sizes did not take part of the training datasets. The colors of the vertical lines match with that of the curves.}
  \label{fig:metrics_training_sizes}
\end{figure}

\section{Predicting NN and NNN distances}
\label{app:target_optimization}

In this section of the Appendix, we check the performance of the GNN in predicting the NN and NNN relative distances simultaneously. Again, we utilize exact training datasets.

Regarding the cost function in Eq.~\eqref{eq:L2_minimization}, one could ask whether one should use both the relative NN and NNN distances whenever both $\chi^{\text{NN}}$ and $\chi^{\text{NNN}}$ are provided, or if one should use an overcomplete set of physical short-range correlators to predict only the relative NN distances, \textit{i.e.} $\chi^{\text{NN}}$ and $\chi^{\text{NNN}}$. As proven in Appendix~\ref{app:bijection_proof}, with arbitrary precision on the correlation functions, one would theoretically only need $\chi^{\text{NN}}$ to predict relative NN distances $\Delta R_{\langle i,j\rangle}$ between Rydberg atoms. However, we may think that giving some extra information to the GNN could be beneficial to overcome numerical or experimental imprecision. To that effect, in Fig.~\ref{fig:metrics_training_targets}, we compare two GNN training cases $\#3$ where, on one hand, only relative NN distances are used as targets, and on the other hand, both NN and NNN relative distances are used as targets, for training sizes $\{4\times 4,5\times 5,6\times 6\}$ and $\{4\times 4,5\times 5,6\times 6,7\times 7\}$. It is very clear that the overcompletion of information helps the GNN palliate the numerical/experimental inaccuracies in determining the relative distances, meaning that predicting solely NN relative distances with $\chi^{\text{NN}}$ and $\chi^{\text{NNN}}$ at hand is better than predicting both NN and NNN relative distances with the same information throughput. It also seems that predictions of larger relative distances (NNN and above) require larger cluster sizes in the training set because edge effects prevail longer. Therefore, if one seeks to also predict NNN relative distances between atoms, one should consider training the GNN with larger cluster sizes to improve substantially the predictive power of the model (compare the navy blue curve with the pink one in Fig.~\ref{fig:metrics_training_targets}). Interestingly, adding another task network specializing in predicting NNN distances worsened the GNN performance compared to using the same task network for predicting NN and NNN distances.

\begin{figure}[h!]
  \centering
    \includegraphics[width=0.95\columnwidth]{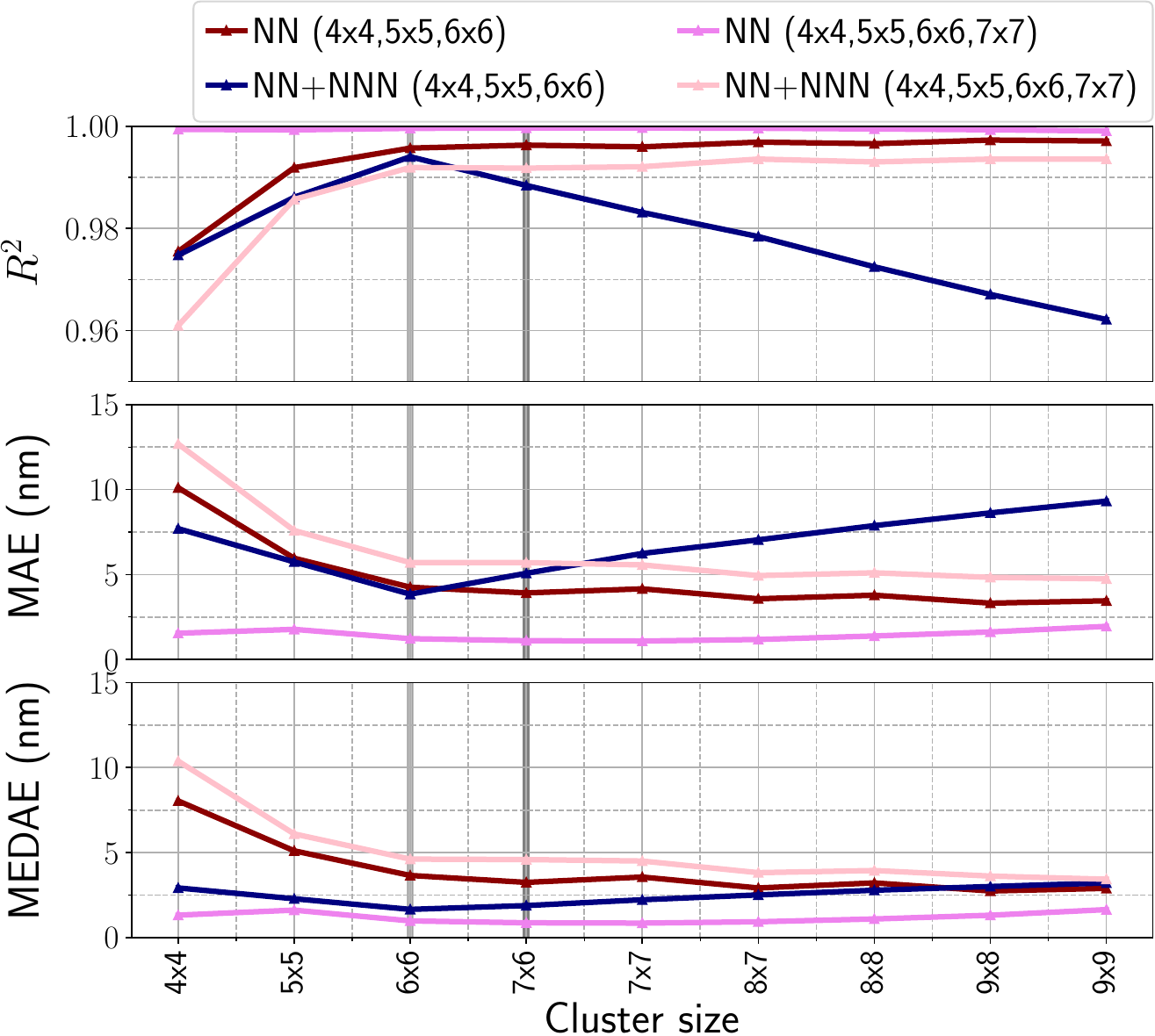}
      \caption{\textbf{Target optimization schemes.} Incidence on various metrics of including the NNN relative distance as a target in the cost function minimization. In blue and light pink, both the NN and NNN relative distances are predicted, while in red and pink, only the NN ones are predicted. Training sizes used are indicated in parentheses. The vertical grey lines refer to the point beyond which sizes lie outside of the training range.}
  \label{fig:metrics_training_targets}
\end{figure}

\section{Details on the numerical implementation and hyperparameters used}
\label{app:technical_details_training_testing}

Some technical details are revealed in this section of the appendix. The GNN was implemented using PyTorch~\cite{paszke_pytorch_2019}, and the data making up the mutually exclusive training and testing datasets were churned out using the DMRG library ITensors.jl~\cite{10.21468/SciPostPhysCodeb.4}.

For all cases laid out in Table~\ref{table:framed}, the same GNN hyperparameters were used. The number of layers defining the PNA was 4, and the number of input features that both nodes and edges corresponded to the length of the $\Omega$-history. The node embedding dimension used was 32. The edge feature dimension corresponds to the combined number of NN edges and, if applicable, NNN edges embedded into a higher dimensional space through a three-layer perceptron. 550 epochs were sufficient to train the GNN and obtain a validation loss comparable to the training one. A histogram of the degree of connectedness of the training dataset graphs is provided to the GNN during the training phase. 
The number of graphs in the training dataset is $2000$ per physical system size, times the $\Omega$-history, making it $20,000$ samples per physical system size. The total number of samples is the same whether the snapshots were used to calculate the correlation function or direct evaluation was employed. The testing dataset comprises $200$ graphs per physical system size times the $\Omega$-history.

Regarding the learning rate profile, the same linearly decreasing learning rate was employed to optimize all the GNNs trained. The starting learning rate value is $5\times 10^{-3}$ while the ending value is $2.5\times 10^{-4}$, for a total of 550 epochs. The learning rate profile plays a big role in the training process: its ramp and extrema should be carefully tuned. The linearly decreasing ramp just mentioned was a good ramp profile.

\section{Technical details of the DMRG calculations}
\label{app:technical_details_DMRG}

To ensure that the Matrix Product States are converged after calculation, the maximum truncation error in the singular value decomposition is capped to single-precision floating-point accuracy ($1\times 10^{-7}$), and we make sure that the energy estimations have converged to seven decimal places with increasing bond dimension through the sweeps. The energy converges within each sweep, and the bond dimension increases from one sweep to another. We monitored how the maximum truncation error and energy estimation changed through 20 DMRG sweeps.  To achieve that, we used a maximal bond dimension of $\chi_{\text{DMRG}}=80$ for a cluster size $4\times 4$, and then increase $\chi_{\text{DMRG}}$ by $10$ for larger cluster sizes, up to $9\times 9$ to get $\chi_{\text{DMRG}}=160$.

\end{document}